\documentclass[]{lmcs}
\pdfoutput=1
\usepackage[utf8]{inputenc}

\usepackage{lastpage}
\lmcsdoi{22}{1}{17}
\lmcsheading{}{\pageref{LastPage}}{}{}%
{Apr.~08,~2024}{Mar.~05,~2026}{}


\usepackage{hyperref}
\usepackage[ruled]{algorithm2e}
\usepackage{wrapfig}
\usepackage{pgfplots}
\usepackage{subcaption}

\makeatletter         % include caption macros from class file to preserve lmcs caption style
\def\figurecaption#1#2{\noindent\hangindent 40pt
                       \hbox to 36pt {\small\sl #1 \hfil}
                       \ignorespaces {\small #2}}
\long\def\@makecaption#1#2{
  \vskip 10pt 
  \settowidth{\@tempdima}{#2}
  \ifdim\@tempdima>0pt
       \setbox\@tempboxa\hbox{#1: #2}
     \else
       \setbox\@tempboxa\hbox{#1 #2}
   \fi
   \ifdim \wd\@tempboxa >\hsize               % IF longer than one line:
       \begin{list}{#1:}{
       \settowidth{\labelwidth}{#1:}
       \setlength{\leftmargin}{\labelwidth}
       \addtolength{\leftmargin}{\labelsep}
        }\item #2 \end{list}\par   % Output in quote mode
     \else                                    %   ELSE  center.
       \hbox to\hsize{\hfil\box\@tempboxa\hfil}  
   \fi}
\makeatother

\newcommand{\off}{\mathsf{offset}}
\newcommand\val{\mathsf{val}}
\newcommand\out{\mathsf{out}}
\newcommand{\denom}{\mathsf{denom}}
\theoremstyle{plain}
\def\eg{{\em e.g.}}
\def\ie{{\em i.e.}}

\begin{document}
\tracingoutput=1
\title[Objective Improvement for Solving Discounted Payoff Games]{An Objective Improvement Approach \texorpdfstring{\\}{} to Solving Discounted Payoff Games\rsuper*}
\titlecomment{{\lsuper*}This work is based on~\cite{DDS23}, which appeared in GandALF'23.}

\author[D.~Dell'Erba]{Daniele Dell'Erba\lmcsorcid{0000-0003-1196-6110}}[a]
\author[A.~Dumas]{Arthur Dumas}[b]
\author[S.~Schewe]{Sven Schewe\lmcsorcid{0000-0002-9093-9518}}[c]
\address{Middlesex University, London, United Kingdom}
\email{d.dellerba@mdx.ac.uk}

\address{ENS Rennes, France}
\email{arthur.dumas@ens.rennes.fr}

\address{University of Liverpool, United Kingdom}
\email{sven.schewe@liverpool.ac.uk}

\begin{abstract}
While discounted payoff games and classic games that reduce to them, like parity and mean-payoff games, are symmetric, their solutions are not.
We have taken a fresh view on the properties that optimal solutions need to have, and devised a novel way to converge to them, which is entirely symmetric.
We achieve this by building a constraint system that uses every edge to define an inequation, and update the objective function by taking a single outgoing edge for each vertex into account.
These edges loosely represent strategies of both players, where the objective function intuitively asks to make the inequation to these edges sharp. In fact, where they are not sharp, there is an `error' represented by the difference between the two sides of the inequation, which is 0 where the inequation is sharp. Hence, the objective is to minimise the sum of these errors.
For co-optimal strategies, and only for them, it can be achieved that all selected inequations are sharp or, equivalently, that the sum of these errors is zero. While no co-optimal strategies have been found, we step-wise improve the error by improving the solution for a given objective function or by improving the objective function for a given solution. 
This also challenges the gospel that methods for solving payoff games are either based on strategy improvement or on value iteration.
\end{abstract}

\maketitle

\section{Introduction}
\label{sec:introduction}

We study turn-based zero sum games played between two players on directed graphs.
The two players take turns to move a token along the
vertices of a finite labelled graph with the goal to optimise their adversarial
objectives.

Various classes of graph games are characterised by the objective of the
players, for instance, in \emph{parity games} the objective is to optimise the parity
of the dominating colour occurring infinitely often, while in \emph{discounted and
  mean-payoff games} the objective of the players is to minimise resp.\ maximise the discounted and limit-average sum of
the colours.  

Solving graph games is the central and most expensive step
in many model
checking~\cite{Koz83,EJS93,Wil01,AHM01,AHK02,SF06},
satisfiability
checking~\cite{Koz83,Var98,Wil01},
and synthesis~\cite{Pit06,SF06a}
algorithms.
Progress in algorithms for solving graph games will, therefore, allow for the
development of more efficient model checkers and contribute to bringing synthesis
techniques to practice.

There is a hierarchy among the graph games mentioned earlier, with simple and well-known reductions from parity games to mean payoff games, from mean-payoff games to discounted payoff games, and from discounted payoff games to simple stochastic games such as the ones from \cite{ZP96}, while no reductions are known in the other direction. Therefore, one can solve instances of all these games by using an algorithm for stochastic games.
All of these games are in \textsf{UP} and \textsf{co-UP} \cite{Jur98}, while no tractable algorithm is known.

Most research has focused on parity games: as the most special class of games, algorithms have the option to use the special structure of their problems, and they are most directly linked to the synthesis and verification problems mentioned earlier.
Parity games have thus enjoyed a special status among graph games and the quest for efficient
algorithms \cite{EL86,EJ91,McN93,ZP96,BCJLM97,Zie98,Obd03,BDM16b,BDM18,BDM18a}  
to solve them has been an active field of research during the last
decades, which has received further boost with the arrival of quasi-polynomial techniques \cite{JL17,FJKSSW19,LB20,LPSW22,DS22,CJKLS22,BDMSW24}.

Interestingly, the only class of efficient techniques for solving parity games that does not (yet) have a quasi-polynomial approach is strategy improvement
algorithms~\cite{Lud95,Pur95,VJ00,BV07,Sch08a,Fea10a,STV15},
a class of algorithms closely related to the Simplex for linear programming, known to perform well in practice. 
Most of these algorithms reduce to mean~\cite{BV07,Sch08a,STV15,BDM20}
or discounted~\cite{Lud95,Pur95,FGO20,Koz21} payoff games.

With the exception of the case in which the fixed-point of discounted payoff games is explicitly computed~\cite{ZP96}, all these algorithms share a disappointing feature: they are inherently non-symmetric approaches for solving an inherently symmetric problem.
However, some of these approaches have a degree of symmetry. Recursive approaches treat even and odd colours symmetrically, one at a time, but they treat the two players very differently for a given colour.
Symmetric strategy improvement \cite{STV15} runs a strategy improvement algorithms for both players in parallel, using the intermediate results of each of them to inform the updates of the other, but at heart, these are still two intertwined strategy improvement algorithms that, individually, are not symmetric.
This is due to the fact that applying strategy improvement itself symmetrically can lead to cycles~\cite{Con93}.

The key contribution of this paper is to devise a new class of algorithms to solve discounted payoff games, which is entirely symmetric.
Like strategy improvement algorithms, it seeks to find co-optimal strategies and improves strategies while they are not optimal.
However, in order to do so, it does not distinguish between the strategies of the two players.
This seems surprising, as maximising and minimising appear to pull in opposing directions.

Similarly to strategy improvement approaches, the new objective improvement approach turns the edges of a game into constraints (here called inequations) and minimises an objective function.
However, while strategy improvement algorithms take only the edges in the strategy of one player (and all edges of the other player) into account and then finds the optimal response by solving the resulting one-player game, objective improvement always takes all edges into account.
The strategies under consideration then form a subset of the inequations, and the goal would be to make them \emph{sharp} (i.e., strict, tight, satisfied as equations), which only works when both strategies are optimal.
When they are not, then there is some \emph{offset} for each of the inequations, and the objective is to reduce this offset in every improvement step.
This treats the strategies of both players completely symmetrically.

\subsection*{Organisation of the Paper}
The paper is organised as follows.
After the preliminaries in Section~\ref{sec:prelims},
we start by outlining our method and using a simple game to explain it in Section~\ref{sec:outline}.
We then formally introduce our objective improvement algorithm in Section~\ref{sec:general}, keeping the question of how to choose better strategies abstract.
Section~\ref{sec:choose} discusses how to find better strategies and Section~\ref{sec:exp} provides an experimental evaluation of the algorithm.
We finally wrap up with a discussion of our results in Section~\ref{sec:discuss}.

\section{Preliminaries}\label{sec:prelims}

A \emph{discounted payoff game} (DPG) is a tuple $\mathcal{G} = (V_{\min}, V_{\max}, E, w, \lambda)$, where $V= V_{\min} \cup V_{\max}$ are the vertices of the game, partitioned into two disjoint sets $V_{\min}$ and $V_{\max}$, such that the pair $(V, E)$ is a finite directed
graph without sinks.
The vertices in $V_{\max}$ (\emph{resp}, $V_{\min}$) are controlled by Player Max or maximiser (\emph{resp}, Player Min or minimiser) and $E \subseteq V \times V$ is the edge relation. Every edge has a weight represented by the function $w : E \to \mathbb{R}$, and a \emph{discount factor} represented by the function $\lambda : E \to [0,1)$. When the discount factor is uniform, i.e., the same for every edge, it is represented by a constant value $\lambda \in [0,1)$.
For ease of notation, we write $w_e$ and $\lambda_e$ instead of $w(e)$ and $\lambda(e)$.
A \emph{play} on $\mathcal{G}$ from a vertex $v$ is an infinite path, which can be represented as a sequence of edges $\rho=e_0 e_1 e_2 \ldots$ such that, for every $i \in \mathbb{N}$, $e_i=(v_i,v_{i+1})\in E$ and $v_0=v$.
By $\rho_i$ we refer to the $i$-th edge of the play.
The \emph{outcome} of a discounted game $\mathcal{G} = (V_{\min}, V_{\max}, E, w, \lambda)$ for a play $\rho$ is $\out(\rho)=\sum_{i=0}^{\infty} w_{e_i} \prod_{j=0}^{i-1}\lambda_{e_j}$.
For games with a constant discount factor, this simplifies into $\out(\rho)=\sum_{i=0}^{\infty} w_{e_i} \lambda^i$.

A \emph{positional strategy} for Max is a function $\sigma_{\max}:V_{\max}\to V$ that maps each Max vertex to a vertex according to the set of edges, i.e., $(v,\sigma_{\max}(v))\in E$.
Positional Min strategies are defined accordingly, and we call the set of positional Min and Max strategies $\Sigma_{\min}$ and $\Sigma_{\max}$, respectively.

A pair of strategies $\sigma_{\min}$ and $\sigma_{\max}$, one for each player Min and Max, defines a unique play $\rho(v,\sigma_{\min},\sigma_{\max})$ from each vertex $v \in V$.
Discounted payoff games are positionally determined~\cite{ZP96}:
$$
\adjustlimits\max_{\sigma_{\max} \in \Sigma_{\max}} \min_{\sigma_{\min} \in \Sigma_{\min}} \out(\rho(v,\sigma_{\min},\sigma_{\max}))
=
\adjustlimits\min_{\sigma_{\min} \in \Sigma_{\min}} \max_{\sigma_{\max} \in \Sigma_{\max}} \out(\rho(v,\sigma_{\min},\sigma_{\max})) 
$$
holds for all $v\in V$, and neither the strategy, nor the value computed from the strategy, changes when we allow more powerful classes of strategies that allow for using memory and/or randomisation for one or both players.

The \emph{value} of a minimiser strategy $\sigma$, denoted by $\mathsf{minival}_{\sigma}: V \rightarrow \mathbb{R}$, is defined as
$$ \mathsf{minival}_{\sigma}: v \mapsto  \max_{\sigma_{\max} \in \Sigma_{\max}}  \out(\rho(v,\sigma,\sigma_{\max}))\ ,$$
The \emph{value} of a maximiser strategy $\sigma$, denoted by $\mathsf{maxival}_{\sigma}: V \rightarrow \mathbb{R}$, is defined as
$$ \mathsf{maxival}_{\sigma}: v \mapsto  \min_{\sigma_{\min} \in \Sigma_{\min}}  \out(\rho(v,\sigma_{\min},\sigma))\ .$$

The resulting \emph{value of $\mathcal G$}, denoted by $\val_{\mathcal G}: V \rightarrow \mathbb{R}$, is defined as
$$ \val_{\mathcal G}: v \mapsto  \max_{\sigma_{\max} \in \Sigma_{\max}} \min_{\sigma_{\min} \in \Sigma_{\min}} \out(\rho(v,\sigma_{\min},\sigma_{\max}))\ .$$

We have $$\val_G(v) = \max_{\sigma_{\max} \in \Sigma_{\max}} \mathsf{maxival}_{\sigma_{\max}}(v) = \min_{\sigma_{\min} \in \Sigma_{\min}} \mathsf{minival}_{\sigma_{\min}} (v)$$
for all $v \in V$. A positional maximiser (resp.\ minimiser) strategy $\sigma$ is said to be \emph{optimal} if, and only if, $\val_{\mathcal G}(v) = w_{(v,\sigma(v))} + \lambda_{(v,\sigma(v))} \val_{\mathcal G}(\sigma(v))$ holds for all maximiser (resp.\ minimiser) vertices.
Such optimal strategies always exist.

The valuation $\val_{\mathcal G}$ is the sole valuation that satisfies the fixed point equations $$\val_{\mathcal G}(v) = \max_{(v,v') \in E} w_{(v,v')} + \lambda_{(v,v')} \val_{\mathcal G}(v')$$ for all maximiser vertices $v \in V_{\max}$ and 
 $$\val_{\mathcal G}(v) = \min_{(v,v') \in E} w_{(v,v')} + \lambda_{(v,v')} \val_{\mathcal G}(v')$$ for all minimiser vertices $v \in V_{\min}$.

Likewise, we define the value of a pair of strategies $\sigma_{\min}$ and $\sigma_{\max}$, denoted $\val_{\sigma_{\min},\sigma_{\max}}: V \rightarrow \mathbb{R}$, as
$$ \val_{\sigma_{\min},\sigma_{\max}}: v \mapsto  \out(\rho(v,\sigma_{\min},\sigma_{\max}))\ .$$

As we treat both players symmetrically in this paper, we define a \emph{pair of strategies} $\sigma : V \mapsto V$ whose restriction to $V_{\min}$ and $V_{\max}$ are a minimiser strategy $\sigma_{\min}$ and a maximiser strategy $\sigma_{\max}$, respectively.
We then write $\rho(v,\sigma)$ instead of $\rho(v,\sigma_{\min},\sigma_{\max})$ and $\val_{\sigma}$ instead of $\val_{\sigma_{\min},\sigma_{\max}}$. To ease the reading, when the strategy is clear from the context, we shall drop the subscript.

If both of these strategies are optimal, we call $\sigma$ a joint \emph{co-optimal} strategy.
This is the case if, and only if, $\val_{\mathcal G}= \val_{\sigma}$ holds.

We define a \emph{solution} to a discounted payoff game $\mathcal G$ to be a valuation $\val$ for the vertices such that, for every edge $e = (v,v')$, it holds that\footnote{These are the constraints represented in $H$ in Section~\ref{sec:general}.}
\begin{itemize}
    \item $\val(v) \leq w_e + \lambda_e \val(v')$ if $v$ is a minimiser vertex and
    \item $\val(v) \geq w_e + \lambda_e \val(v')$ if $v$ is a maximiser vertex.
\end{itemize}
For co-optimal strategies $\sigma$, $\val_{\sigma}$ that is equal to $\val_{\mathcal G}$ is clearly a solution, as all of these inequations are satisfied, while for a pair of strategies $\sigma$ that is not co-optimal, $\val_{\sigma}$ is not a solution%
\footnote{For every joint strategy $\sigma$ of both players and every vertex $v\in V$, we have $\val_{\sigma}(v) = w_{(v,\sigma(v))} + \lambda_{(v,\sigma(v))} \val_{\sigma}(\sigma(v))$ by definition.
Thus, if $\val_{\sigma}$ is a solution, then $\val_{\mathcal G}$ also satisfies the fixed point equations that define $\val_{\mathcal G}$, and we have that $\val_\sigma=\val_{\mathcal G}$ holds and that $\sigma$ is co-optimal.
}.

Note that we are interested in the \emph{value} of each vertex, not merely if the value is greater or equal than a given threshold value.

\subsection{Simplex method}

A common way to compute $\val_{\mathcal G}$ is to create a sequence of linear programming instances from $\mathcal G$. In this paragraph, we provide the background for solving a linear programming instance and outline the Simplex method that is the standard algorithm used to find a solution. 

A linear programming problem requires to find a (solution) vector $x\in\mathbb{R}^{n}$ that minimises (or maximises) $c\cdot x$, where $c\in \mathbb{R}^{n}$ is a vector of constants, such that $A\cdot x \leq b$ (or $\geq b$), where $A\in \mathbb{R}^{m\times n}$ is a coefficient matrix and $b\in \mathbb{R}^{m}$ is a vector of constants. 
Given a game $\mathcal G$, we can iteratively define an objective function represented by the vector c, the input matrix $A$, and the vector $b$, such that, in the sequence of linear programming instances, the solution $x$ converges to $\val_{\mathcal G}$. In this translation from $\mathcal G$ to the linear programming instance, the matrix $A$ is composed of $m$ rows and $n$ columns where $m$ is the number of edges and $n$ the number of vertices in $\mathcal G$. Since the constrains in $A$ represent the edges, that are binary relations, only two columns in each row contain a nonzero value (with the exception of self-loops that have only one nonzero coefficient). The vector $b$ contains the weight of the edges. Hence, every row $i$ can be expressed as $a_{ij}\cdot x_j+a_{ik}\cdot x_k\geq b_i$, if $e=(j,k)$ is an edge of $\mathcal G$ and $j$ is a maximiser vertex (otherwise the constraint is $\leq b_i$), with $a_{ij}=1$, $a_{ik}=-\lambda_e$, and $b_i=w_e$. Finally, $c$ is the unit vector with a positive sign for the maximiser vertices and negative sign for the minimiser, and the objective function maximises the sum of the values of $x$.

As we use no details of any implementation of a simplex algorithm, we only rehash the principle way it operates here.
From a graphical point of view, the set of constraints, i.e., the rows of $A$ and $b$, define a region called \emph{polytope} (or: simplex; hence the name of the approach) in the space $\mathbb{R}^{n}$. Any point that lies within the polytope (including its surface) is a solution.

Broadly speaking, the simplex algorithm traverses the corner of the polytope so that the value the linear objective function takes for the corner is improving -- in our case decreasing as we minimise.
More precisely, it picks subsets of $n$ inequations that are made sharp:
when read as equations, they define a valuation in $\mathbb R^n$ that satisfies all $m$ inequations; they are bases and define the corners of the simplex.

For this, the simplex algorithm operates in two phases. In a first phase it finds \emph{any} basis that defines such a corner. We will not discuss this phase as it is not relevant for this article%
\footnote{It is also the case that, in all but the first call to a linear program, we already have a basis to start from.}.

In the second phase, the simplex algorithm essentially updates the basis to a \emph{neighbouring} one, meaning that only one inequation is removed from and one inequation is entered into the set of sharp inequations.
When making this pick, the simplex algorithm only allows for such updates that are (1) still bases (i.e.\ their rows in $A$ must be linearly independent to that they define a point in $\mathbb{R}^n$); (2) the point in $\mathbb{R}^n$ they define satisfies all $m$ inequations, and (3) the value of the objective function improves.

For non-optimal solutions (w.r.t.\ the objective function), this is always possible, though in some cases the improvement in (3) is not strict. Even in those cases, there is still a series of changes that all satisfy (1-3) to an optimal solution.
Those cases are called degenerate, that is, when the current solution satisfies more inequations than necessary and, although there exists a neighbouring solution leading to a strict improvement, picking it requires a basis change. While we do not discuss how they are treated in simplex algorithms, we point out that the problems and their solutions are similar to what we describe in Sections \ref{ssec:noLocal} and \ref{ssec:efficient}, respectively.

\section{Outline and Motivation Example}\label{sec:outline}

We start by considering the simple discounted payoff game of Figure~\ref{fig:example}, assuming that it has some uniform discount factor $\lambda \in [0,1)$.
In this game, the minimiser (who owns the right vertex, depicted as a square) has only one option: she always has to use the self-loop, which earns her an immediate reward of $1$.
The overall reward the minimiser reaps for a play that starts in her vertex is therefore 
 $1 + \lambda + \lambda^2 + \ldots = \frac{1}{1-\lambda}$.
 
The maximiser (who owns the left vertex, $b$, marked by a circle) can choose to either use the self-loop, or to move to the minimiser vertex (marked by a square), both yielding no immediate reward.

If the maximiser decides to stay forever in his vertex (using the self-loop), his overall reward in the play that starts at (and, due to his choice, always stays in) his vertex, is $0$.
If he decides to move on to the minimiser vertex the $n^{th}$-time, 
then the reward is $\frac{\lambda^n}{1-\lambda}$.

The optimal decision of the maximiser is, therefore, to move on the first time, which yields the maximal reward of $\frac{\lambda}{1-\lambda}$.
Every vertex $v$ has some outgoing edge(s) $e = (v,v')$ where $\val_{\mathcal G}(v) = w_e + \lambda_e \val_{\mathcal G}(v')$ holds~\cite{ZP96};
these edges correspond to the optimal decisions for the respective player.

For our running example game of Figure~\ref{fig:example} with a fixed discount factor $\lambda \in [0,1)$, the constraints that a solution must satisfy are:
\begin{enumerate}
    \item $\val(a) \leq 1 + \lambda \val(a)$ \hspace{0.15em} for the self-loop of the minimiser vertex;
    \item $\val(b) \geq \lambda \val(b)$ \hspace{2.1em} for the self-loop of the maximiser vertex; and
    \item $\val(b) \geq \lambda \val(a)$ \hspace{2em} for the transition from the maximiser to the minimiser vertex.
\end{enumerate}

\begin{wrapfigure}{r}{0.35\textwidth} 
\vspace{-21pt}
  \begin{center}
\begin{tikzpicture}
    [node distance = 5em, bend angle = 22.5, inner sep = 0.3em, minimum size = 3.2em]
    \tikzset{every loop/.style = {max distance = 1.5em}}
    \node[shape=circle,draw,thick] (max) at (0,0) {b};
    \node[shape=rectangle,draw,thick] (min) at (3,0) {a};

    \draw[thick,->] (max) to node [above] {$0$} (min);
    \draw [thick,->] (max) edge[loop above]node{$0$} (max);
    \draw [thick,->] (min) edge[loop above]node{$1$} (min);
    \end{tikzpicture}
\caption{A discounted Payoff Game. Maximiser vertices are depicted as a circle, minimizer ones as a square.}
\label{fig:example}
  \end{center}
\end{wrapfigure}
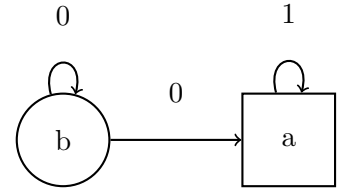
The unique solution that satisfies these inequations and produces a sharp inequation (i.e.\ satisfied as equation) for some outgoing edge of each vertex assigns $\val(a)=\frac{1}{1-\lambda}$ and $\val(b)=\frac{\lambda}{1-\lambda}$.
This solution also defines the optimal strategies of the players (to stay for the minimiser, and to move on for the maximiser).

Solving a discounted payoff game means finding this solution and/or these strategies.

We discuss a symmetric approach to find this unique valuation.
Our approach adjusts linear programming in a natural way that treats both players symmetrically:
we maintain the set of inequations for the complete time, while approximating the goal of ``one equality per vertex'' by the objective function.
To do that, we initially fix an \emph{arbitrary} outgoing edge for every vertex (a strategy), and minimise the sum of the distances between the left and right sides of the inequations defined by these edges, which we call the \emph{offset} of this edge. This means, for an edge $e = (v,v')$, to minimise the difference between $\val(v)$ (left side of the inequation) and $w_e + \lambda_e \val(v')$ (right side).

To make this clear, we consider again the example of Figure~\ref{fig:example} and use both self-loops as the strategies for the players fixed at the beginning in our running example.
The offset for the selected outgoing edge of the minimiser vertex $a$ is equal to $1 -(1-\lambda)\val(a)$, while the offset for the selected outgoing edge of the maximiser vertex $b$ is equal to $(1-\lambda)\val(b)$.
The resulting overall objective consists, therefore, in minimising the value $1-(1-\lambda)\val(a) + (1-\lambda)\val(b)$.

This term is always non-negative, since it corresponds to the sum of the edges' contributions that are all non-negative.
Moreover, when only optimal strategies are selected to form this objective function, the value $0$ can be taken, and where it is taken, it defines the correct valuation of the game.

As the maximiser's choice to take the self-loop is not optimal, the resulting objective function the strategies define, that is $1-(1-\lambda)\val(a) + (1-\lambda)\val(b)$, cannot reach $0$.
But let us take a look at what an optimal solution w.r.t.\ this objective function looks like.

Optimal solutions can be taken from the corners of the polytope defined by the inequations, also known as its simplex.
In this case, the optimal solution (w.r.t.\ this initial objective function) is defined by making inequations (1) and (3) sharp: this provides the values $\val(a)=\frac{1}{1-\lambda}$ and $\val(b)=\frac{\lambda}{1-\lambda}$; the objective function takes the value $\lambda$ at this point.

For comparison, in the other corner of the simplex, defined by making inequations (2) and (3) sharp, we obtain the values $\val(a)=\val(b)=0$; the objective function takes the value $1$ at this point.
Finally, if we consider the last combination, making (1) and (2) sharp, this provides the values $\val(a)=\frac{1}{1-\lambda}$ and $\val(b)=0$, so that inequation (3) is not satisfied; this is therefore not a corner of the simplex.

Thus, in this toy example, while selecting the wrong edge cannot result in the objective function taking the value $0$, we still found the optimal solution.
In general, we might need to update the objective function.
To update the objective function, we change the outgoing edges of some (or all) vertices so that the overall value of the objective function goes down.
Note that this can be done not only when the linear program returns an optimal solution but also during its computation.
For example, when using a simplex method, updating the objective function can be used as an alternative pivoting rule at any point during the traversal of the simplex.

In general, the valuation returned as solution is computed by using objective functions based on strategies that are not necessarily optimal.
We now adjust the example so that choosing both self-loops is not only non-optimal, but also does not define the correct valuation of the DPG. 
To this end, we use different discount factors%
\footnote{Note that we can also replace the transitions with a smaller discount factor by multiple transitions with a larger discount factor. This would allow for keeping the discount factor uniform, but needlessly complicate the discussion and inflate the size of the example.}
for the game of Figure~\ref{fig:example}: we choose $\frac{1}{3}$ for the self-loop of the maximiser vertex and $\frac{2}{3}$ for the other two transitions, so that the three resulting inequations are (1) $\val(a) \leq 1 + \frac{2}{3} \val(a)$; (2) $\val(b) \geq \frac{1}{3} \val(b)$; and (3) $\val(b) \geq \frac{2}{3} \val(a)$.
Choosing both self loops results in the objective function to minimise $1-\frac{1}{3}\val(a) + \frac{2}{3}\val(b)$.

Making the adjusted inequations (2) and (3) sharp still results in the values $\val(a)=\val(b)=0$, and the objective function still takes the value of $1$. While making inequations (1) and (3) sharp provides the values $\val(a)=3$ and $\val(b)=2$; the objective function takes the value $\frac{4}{3}$ at this point.
Finally, if we consider the last combination, making (1) and (2) sharp still conflicts with inequation (3).

Thus, $\val(a)=\val(b)= 0$ would be the optimal solution for the given objective function, which is not the valuation of the game. 
We will then update the candidate strategies so that the sum of the offsets goes down.

\subsection{Comparison with strategy improvement}

The closest relatives to our new approach are strategy improvement algorithms.
Classic strategy improvement approaches solve the problem of finding the valuation of a game (and usually also co-optimal strategies) by (1) fixing a strategy for one of the players (we assume w.l.o.g. that this is the maximiser), (2) finding a valuation function for the one player game that results from fixing this strategy (often together with an optimal counter strategy for their opponent), and (3) updating the strategy of the maximiser by applying local improvements.
This is repeated until no local improvements are available, which entails that the constraint system is satisfied.

\begin{figure}[ht]
  \centering
  \begin{minipage}{.6\linewidth}
    \begin{algorithm}[H]
      \caption{\label{alg:si_alg} Strategy Improvement}
      \SetKwInOut{Input}{input} \SetKwInOut{Output}{output}
      
      \Input{A discounted payoff game $\mathcal{G} = (V_{\min}, V_{\max}, E, w, \lambda)$}
      \Output{The valuation $\val$ of $\mathcal G$}
      \SetInd{0.25em}{0.5em}
        {
        \nl $f \leftarrow{} \mathsf{ObjectiveFunction}(\mathcal{G})$ \\
        \nl $\sigma \leftarrow{} \mathsf{ChooseInitialStrategy}(\mathcal{G})$\\
        \nl \While{\sf{true}}{
            \nl $H \leftarrow{} \mathsf{Inequations}(\mathcal{G},\sigma)$\\
            \nl $\val \leftarrow{} \mathsf{LinearProgramming}(H, f)$\\
            \nl \If{$\sigma$ is optimal}{\Return $\val$}
                \nl $\sigma \leftarrow{} \mathsf{ChooseBetterStrategy}(\mathcal{G},\val,\sigma)$
            }
        }
    \end{algorithm}
    \end{minipage}
\end{figure}

These steps can be identified in Algorithm~\ref{alg:si_alg}. The objective function is fixed at the beginning (Line 1) and corresponds to the sum of the values of all the vertices. The algorithm simply maximises it%
\footnote{The optimal minimiser strategy chooses edge that minimises the value for every minimiser vertex. In the linear program, this is revlected by an inequeation for every outgoing edge, like the inequation $\val(a) \leq 1 + \lambda \val(a)$ from the running example, and then maximises the value, which forces the selected value to be the minimal value such that one such inequation is sharp.}
by searching for a solution that maximises the value of all vertices. (1) The strategy of one of the two players (for convention the maximiser -- if we want to instead fix the strategy of the minimiser, we would have to minimise the sum of the values for all vertices instead) is initialised at Line 2. (2) Before the computation of the valuation at Line 5, the set of constraints is updated accordingly to the selected strategy of the maximiser (Line 4). At this point, if the strategy is optimal w.r.t.\ the current valuation, then the algorithm ends. Otherwise, (3) there exists an improvement, which is applied at Line 7.

For Step (2) of this approach, we can use linear programming, which does invite a comparison with our technique.
The linear program for solving Step (2) would not use all inequations: it would, instead, replace the inequations defined by the currently selected edges of the maximiser by equations, while dropping the inequations defined by the other maximiser transitions.
The objective function would then be to maximise the values of all vertices while still complying with the remaining (in)equations.

Thus, in our novel symmetric approach, the constraints remain while the objective is updated; in strategy improvement, instead, the objective remains while the constraints are updated.
Moreover, the players and their strategies are treated quite differently in strategy improvement algorithms: while the candidate strategy of the maximiser results in making the inequations of the selected edges sharp (and dropping all other inequations of maximiser edges), the optimal counter strategy is found by maximising the objective. This is again in contrast to our novel symmetric approach, which treats both players equally.

A small further difference is in the valuations that can be taken: the valuations that strategy improvement algorithms can take are the valuations of strategies, while the valuations our objective improvement algorithm can take on the way are the corners of the simplex defined by the inequations.
Except for the only intersection point between the two (the valuation of the game), these corners of the simplex do not relate to the value of strategies. 
Table~\ref{tab:comparison} summarises these observations.

\renewcommand{\arraystretch}{0.9}
\begin{table}[t]
    \centering
    \begin{tabular}{c||c|c}
         & \textbf{Objective Improvement} & \textbf{Strategy Improvement} \\[0.5em]
        \hline
        \rule{0pt}{1em}
        \textbf{players}  & symmetric treatment & asymmetric treatment \\[0.5em]
        \textbf{constraints} & remain the same: & change: \\
         & one inequation per edge & one inequation for each edge \\
         & & defined by the current strategy for \\
         & & the strategy player, one inequation \\
         & & for every edge of their opponent \\[0.5em]
        \textbf{objective} & minimise errors for selected edges & maximise values \\[0.5em]
        \textbf{update} & objective: & strategy: \\
         & one edge for each vertex & one edge for each vertex \\
         & & of the strategy player  \\[0.5em]
        \textbf{valuations} & corners of simplex & defined by strategies
         
    \end{tabular}
    \vspace{1em}
    \caption{A comparison of the novel objective improvement with classic strategy improvement.}
    \label{tab:comparison}
\end{table}

\section{General Objective Improvement}\label{sec:general}

In this section, we present the approach outlined in the previous section more formally, while keeping the most complex step -- updating the candidate strategy to one which is \emph{better} in that it defines an optimisation function that can take a smaller value -- abstract. (We turn back to the question of how to find better strategies in Section~\ref{sec:choose}.)
This allows for discussing the principal properties more clearly.

\begin{figure}[ht]
  \centering
  \begin{minipage}{.6\linewidth}
    \begin{algorithm}[H]
      \caption{\label{alg:oi_alg} Objective Improvement}
      \SetKwInOut{Input}{input} \SetKwInOut{Output}{output}
      
      \Input{A discounted payoff game $\mathcal{G} = (V_{\min}, V_{\max}, E, w, \lambda)$}
      \Output{The valuation $\val$ of $\mathcal G$}
      \SetInd{0.25em}{0.5em}
        {
        \nl $H \leftarrow{} \mathsf{Inequations}(\mathcal{G})$\\
        \nl $\sigma \leftarrow{} \mathsf{ChooseInitialStrategies}(\mathcal{G})$\\
        \nl \While{\sf{true}}{
            \nl $f_\sigma \leftarrow{} \mathsf{ObjectiveFunction}(\mathcal{G}, \sigma)$ \\
            \nl $\val \leftarrow{} \mathsf{LinearProgramming}(H, f_\sigma)$\\
            \nl \If{ $f_\sigma(\val) = 0$}{\Return $\val$}
                \nl $\sigma \leftarrow{} \mathsf{ChooseBetterStrategies}(\mathcal{G},\val,\sigma)$
            }
        }
    \end{algorithm}
    \end{minipage}
\end{figure}

A general outline of our \emph{objective improvement} approach is reported in Algorithm~\ref{alg:oi_alg}.
Before describing the procedures called by the algorithm, we first outline the principle.

When running on a discounted payoff game $\mathcal{G}= (V_{\min}, V_{\max}, E, w, \lambda)$, the algorithm uses a set of inequations defined by the edges of the game and the owner of the source of each edge.
This set of inequations denoted by $H$ contains one inequation for each edge and (different to strategy improvement approaches whose set of inequations is a subset of $H$) $H$ never changes.

The inequations from $H$ are computed by a function called $\mathsf{Inequations}$ (Line 1)  that, given the discounted game $\mathcal{G}$, returns the set made up of one inequation per edge $e=(v,v') \in E$, defined as follows:
\[ I_e =  \begin{cases}
        \val(v) \geq w_{e} + \lambda_e \val(v') &\text{if }v \in V_{\max}, \\
        \val(v) \leq w_{e} + \lambda_e \val(v') &\text{otherwise.}
    \end{cases} 
\]

The set $H = \{I_e \mid e \in E\}$ is defined as the set of all inequations for the edges of the game.

The algorithm also handles strategies for both players, treated as a single strategy $\sigma$.
They are initialised (for example randomly) by the function $\mathsf{ChooseInitialStrategies}$ at Line 2.

This joint strategy is used to define an objective function $f_{\sigma}$ by calling function $\mathsf{ObjectiveFunction}$ at Line 4, whose value on an evaluation $\val$ is: $f_{\sigma}(\val)=\sum_{v \in V} f_\sigma(\val, v)$ with the following objective function components:
\[
  f_\sigma(\val, v) = \off(\val,(v,\sigma(v)))
\]
where the offset of an edge $(v,v')$ for a valuation is defined as follows:
\[
  \off(\val, (v,v')) =
  \begin{cases}
    \val(v) - (w_{(v,v')} + \lambda_{(v,v')} \val(v')) &\text{if } v \in V_{\max}, \\
    (w_{(v,v')} + \lambda_{(v,v')} \val(v')) - \val(v) &\text{otherwise.}
  \end{cases}
\]

This objective function $f_\sigma$ is given to a linear programming algorithm, alongside with the inequations set $H$.
We underline that, due to the inequation $I_{(v,v')}$, the value of $\off(\val, (v,v'))$ is non-negative for all $(v,v')\in E$ in any solution $\val$ (optimal or not) that satisfies the system of inequations $H$.

\begin{obs}
\label{obs:nonneg}
At Line 6 of Algorithm~\ref{alg:oi_alg}, the value of $f_\sigma(\val)$ is non-negative.
\end{obs}

We put a further restriction on $\val$, returned by the call $\mathsf{LinearProgramming}$, in that we require it to be the solution to a \emph{basis} $\mathbf{b}$ in $H$.
Such a basis consists of $|V|$ inequations that are satisfied sharply (again, as equations), such that, according to the definition of basis, these $|V|$ equations uniquely define the values of all vertices.
We refer to this unique set of values as \emph{the evaluation of $\mathbf{b}$}, denote by $\val_{\mathbf{b}}$.
Note that $\val_{\mathbf{b}}$ is in particular a solution of $H$, so that all $|E|$ inequations are satisfied, condition that in general does not holds.

The call $\mathsf{LinearProgramming}(H, f_\sigma)$ to some linear programming algorithm returns a solution $\val$ of the vertices that minimises $f_\sigma$ while satisfying $H$ (Line 5); for convenience, we require this solution to also be the evaluation $\val_{\mathbf{b}}$ for some basis $\mathbf{b}$ of $H$.
(Note that the simplex algorithm, for example, only uses solutions of this form in every step.)
We call this solution a \emph{solution associated to $\sigma$}.

We say that a solution $\val$ \emph{defines} strategies of both players if, for every vertex $v\in V$, the inequation of (at least) one of the outgoing edges of $v$ is sharp.
These are the strategies defined by using, for every vertex $v \in V$, an outgoing edge for which the inequation is sharp.
Note that there can be more than one of these inequations for some of the vertices.

\begin{obs}
\label{obs:0strategy}
If, for a solution $\val$ of $H$, $f_\sigma(\val) = 0$ holds, then, for every vertex $v \in V$, the inequation $I_{(v,\sigma(v))}$ for the edge $(v,\sigma(v))$ is sharp and $\val$ therefore defines strategies for both players -- those defined by $\sigma$, for example.
\end{obs}

Instead of checking whether $\val$ defines the strategies of both players, we can use, alternatively, $f_\sigma(\val) = 0$ as a termination condition, as shown at Line 6 in Algorithm~\ref{alg:oi_alg}, since in this case $\sigma$ must define co-optimal strategies.

\begin{thm}
\label{thm:main}If $\sigma$ describes co-optimal strategies, then $f_\sigma(\val) = 0$ holds at Line 6 of Algorithm~\ref{alg:oi_alg}.
If $\val$ (from Line 5 of Algorithm~\ref{alg:oi_alg}) defines joint strategies $\sigma'$ for both players, then $\sigma'$ is co-optimal and $\val$ is the valuation of $\mathcal G$.
\end{thm}

\begin{proof}
The valuation $\val=\val_{\mathcal G}$ of the game is the unique solution of $H$ for which, for all vertices $v$, the inequation to (at least) one of the outgoing edges of $v$ is sharp. This case implies that the edges for which they are sharp describe co-optimal strategies. 
This solution of $H$ is thus the only one that \emph{defines} strategies for both players, which shows the second claim. The uniqueness of the solution is a consequence of the sharpness: while inequalities define the solution region, equations select exactly one.

Moreover, if $\sigma$ describes co-optimal strategies, then $f_\sigma (\val)=0$ holds for $\val=\val_{\mathcal G}$ (and for this valuation only), which establishes the first claim.
\end{proof}

The theorem above ensures that, in case the condition at Line 6 holds, the algorithm terminates and provides the value of the game that then allows us to infer optimal strategies of both players.
Otherwise, we have to improve the objective function and make another iteration of the while loop.
At Line 7, $\mathsf{ChooseBetterStrategies}$ can be any procedure that, for $f_\sigma(\val) \neq 0$, provides a joint strategy $\sigma'$ \emph{better} than $\sigma$ as defined below.

\begin{defi}[Better Strategy]
\label{def:better}
    A joint strategy $\sigma'$ for both players is \emph{better} than a joint strategy $\sigma$ if, and only if,
    $$\min_{\val'} f_{\sigma'}(\val') < \min_{\val} f_\sigma(\val)$$
\end{defi}

Given a game $\mathcal{G}$, a strategy $\sigma'$ and the corresponding objective function $f_{\sigma'}$ computed by $\mathsf{ObjectiveFunction}(\mathcal{G}, \sigma')$, the call $\mathsf{LinearProgramming}(H, f_{\sigma'})$ computes the minimal value of $f_{\sigma'}$. If this value is strictly lower than the minimal value for $f_\sigma$, computed by $\mathsf{LinearProgramming}(H, f_{\sigma})$, then $\sigma'$ is better than $\sigma$.

While we discuss how to implement this key function in the next section, we observe here that the algorithm terminates with a correct result with any implementation that chooses a better objective function in each round: correctness is due to it only terminating when $\val$ \emph{defines} strategies for both players, which implies (cf.\ Theorem~\ref{thm:main}) that $\val$ is the valuation of $\mathcal G$ ($\val=\val_{\mathcal G}$) and all strategies defined by $\val$ are co-optimal.
Termination is obtained by a finite number of positional strategies: by Observation~\ref{obs:nonneg}, the value of the objective function of all of them is non-negative, while the objective function of an optimal solution to co-optimal strategies is $0$ (cf. Theorem~\ref{thm:main}), which meets the termination condition of Line 6 (cf.\ Observation~\ref{obs:0strategy}).

\begin{cor}
Algorithm~\ref{alg:oi_alg} always terminates with the correct value.
\end{cor}

\section{Choosing Better Strategies}
\label{sec:choose}
In this section, we discuss sufficient criteria for finding a procedure that efficiently implements $\mathsf{ChooseBetterStrategies}$.
For this, we make four observations described in the next subsections:

\begin{enumerate}
    \item All local improvements can be applied.
        A strategy $\sigma'$ is a local improvement to a strategy $\sigma$ if $f_{\sigma'}(\val) < f_\sigma(\val)$ holds for the current solution $\val$ (Section~\ref{ssec:local}).
    
    \item If the current solution $\val$ does not \emph{define} a pair of strategies $\sigma$ for both players and has no local improvements, then a better strategy $\sigma'$ can be found applying only switches from and to edges that already have offset $0$ (Section~\ref{ssec:noLocal}).
    
    \item The improvement mentioned in the previous point can be found for special games -- the sharp and improving games defined in Section~\ref{ssec:efficient} -- by trying a single-edge switch.
    
    \item Games can almost surely be made sharp and improving by adding random noise that retains optimal strategies (Section~\ref{ssec:sharpeing}).
\end{enumerate}

Together, these four points provide efficient means for finding increasingly better strategies, and thus to find the co-optimal strategies and the valuation of the discounted payoff game.

As a small side observation, when using a simplex-based technique to implement $\mathsf{LinearProgramming}$ at Line 5 of Algorithm~\ref{alg:oi_alg}, then the pivoting of the objective function from point (1) and the pivoting of the basis can be mixed (this will be discussed in Section~\ref{ssec:mixing}).

\subsection{Local Improvements}\label{ssec:local}

The simplest and most common case of creating better strategies $\sigma'$ from a solution for the objective $f_\sigma$ for a strategy $\sigma$ is to consider \emph{local improvements}.
Much like local improvements in strategy iteration approaches, local improvements consider, for each vertex $v$, a successor $v' \neq \sigma(v)$, such that $\off(\val,(v,v')) < \off(\val,(v,\sigma(v)))$ for the current solution $\val$, which is optimal for the objective function $f_\sigma$.

To be more general, our approach does not necessarily require one to select only local improvements, but it can work with global improvements, though we cannot see any practical use of choosing differently.
For example, if we treat the function as a global improvement approach, we can update the value of a vertex $v$ so that it increases by 1 and update the value of another vertex $v'$ so that it decreases by 2. The overall value of the function will decrease even if locally some components increased their value. Interestingly, this cannot be done with a strategy improvement approach, as it requires one to always locally improve the value of each vertex when updating.

\begin{lem}\label{lem:better}
If $\val$ is an optimal solution for the linear programming problem at Line 5 of Algorithm~\ref{alg:oi_alg} and $f_{\sigma'}(\val) < f_\sigma(\val)$, then $\sigma'$ is better than $\sigma$.
\end{lem}

\begin{proof}
The solution $\val$ is, being an optimal solution for the objective $f_\sigma$, a solution to the system of inequations $H$. For a solution $\val'$ that is optimal for $f_{\sigma'}$, we thus have $f_{\sigma'}(\val') \leq f_{\sigma'}(\val) < f_\sigma(\val)$, which implies that $\sigma'$ is better than $\sigma$ according to Definition~\ref{def:better} of a better strategy.
\end{proof}

This allows us to identify local improvements cheaply.

\begin{cor}\label{cor:better}
If $\val$ is an optimal solution for the linear programming problem at Line 5 of Algorithm~\ref{alg:oi_alg}, there is an edge $(v,v')$ such that $\off(\val,(v,v')) < \off(\val,(v,\sigma(v)))$, and $\sigma'$ is a strategy for both players s.t.\ $\off(\val,(v,v')) \geq \off(\val,(v,\sigma'(v)))$ holds for all edges $(v,v')$, then $\sigma'$ is better than $\sigma$.
\end{cor}

We say that $\val$ \emph{identifies} these strategies.

\subsection{No Local Improvements}\label{ssec:noLocal}

The absence of local improvements means that, for all vertices $v \in V$ and all outgoing edges $(v,v') \in E$, $\off(\val,(v,v')) \geq \off(\val,(v,\sigma(v)))$.

We define for a solution $\val$ optimal for $f_\sigma$ (like the $\val$ produced in Line 5 of Algorithm~\ref{alg:oi_alg}):

\begin{itemize}
\item $S_\val^\sigma = \{ (v,v') \in E \mid \off(\val,(v,v')) \leq \off(\val,(v,\sigma(v)))\}$ as the set of \emph{at least stale} edges;
    naturally, every vertex has at least one outgoing stale edge: the one defined by $\sigma$;
    
\item $E_\val = \{ (v,v') \in E \mid \off(\val,(v,v')) =  0\}$ as the set of edges, for which the inequation for $\val$ is sharp;
in particular, all edges in the basis of $H$ that defines $\val$ are sharp (and at least stale); and
    
\item $E_\val^\sigma$ as any set of edges between $E_\val$ and $S_\val^\sigma$ (i.e.\ $E_\val \subseteq E_\val^\sigma \subseteq S_\val^\sigma$) such that $E_\val^\sigma$ contains an outgoing edge for every vertex $v \in V$;
   we are interested to deal with sets that retain the game property that every vertex has a successor, we can do that by adding (non-sharp) at least stale edges to $E_\val$.
\end{itemize}  

Note that $S_\val^\sigma$ is such a set, and, therefore, an adequate set is easy to identify.
However, we might be interested in keeping the set small and choosing the edges defined by $E_\val$ plus one outgoing edge for every vertex $v$ that does not have an outgoing edge in $E_\val$.
Where there is no local improvement, the most natural solution is to choose the edge $(v,\sigma(v)) \in E_\val^\sigma$ defined by $\sigma$ for each such vertex $v$.

\begin{obs}\label{obs:stalegame}
If $\mathcal{G}= (V_{\min}, V_{\max}, E, w, \lambda)$ is a DPG and $\sigma$ a strategy for both players such that $\val$ is an optimal solution for the objective $f_\sigma$ to the system of inequations $H$,
then $\mathcal G'= (V_{\min}, V_{\max}, E_\val^\sigma, w, \lambda)$ is also a DPG.
\end{obs}

This simply holds because every vertex $v \in V$ retains at least one outgoing transition.

\begin{lem}\label{lem:stale}
Let $\mathcal{G}= (V_{\min}, V_{\max}, E, w, \lambda)$ be a DPG, $\sigma$ a strategy for both players, $\val$ an optimal solution returned at Line 5 of Algorithm~\ref{alg:oi_alg} for $f_\sigma$.
If $\val$ does not define strategies of both players, then there is a better strategy $\sigma'$ such that, for all $v \in V$, $(v,\sigma'(v)) \in E_\val^\sigma$.
\end{lem}
\begin{proof}
By Observation~\ref{obs:stalegame},  $\mathcal G'=  (V_{\min}, V_{\max}, E_\val^\sigma, w, \lambda)$ is a DPG. Let $\val'$ be the value of $\mathcal G'$, and $\sigma'$ be the strategies for the two players defined by it.

If $\val'$ is also a solution of $\mathcal G$, then we are done. However, this need not be the case, as the set of inequations $H'$ for $\mathcal G'$ is smaller than the set of inequations $H$ for $\mathcal G$, so $\val'$ might violate some of the inequations that are in $H$, but not in $H'$.
Given that $\val'$ is a solution for $\mathcal G'$, it satisfies all inequations in $H'$.
Moreover, since $\val$ also satisfies all inequations of $H'$, it follows that the same inequations hold for every convex combination of $\val$ and $\val'$.

We now note that the inequations of $H$ that are not in $H'$ are not sharp for $\val$. Thus, there is an $\varepsilon \in (0,1]$ such that the convex combination $\val_\varepsilon = \varepsilon \cdot \val' + (1-\varepsilon) \cdot \val$ is a solution to those inequations.

We have $\off(\val,(v,\sigma'(v))) \leq \off(\val,(v,\sigma(v)))$ for every vertex $v$, as $\sigma'$ only uses the at least stale edges from $E^\sigma_\val$, so that in particular $f_{\sigma'}(\val) \leq f_\sigma(\val)$ holds.
We also have $f_\sigma(\val)>0$ by assumption and $f_{\sigma'}(\val')=0$ by the optimality of $\val'$ for $H'$.

For an optimal solution $\val''$ of $H$ for the objective $f_{\sigma'}$, this provides $f_{\sigma'}(\val'') \leq f_{\sigma'}(\val_\varepsilon) < f_\sigma(\val)$.

Therefore, $\sigma'$ is better than $\sigma$.
\end{proof}

Where $\val$ does not define strategies of both players, there are always at least $|V|$ inequations sharp, and thus there are vertices with multiple outgoing edges in $E^\sigma_\val$ that all have $0$ offset.

While Lemma \ref{lem:stale} also holds when there is a local improvement, we would not use it then, as finding local improvements is much easier.

When there is no local improvement, the most natural choice is $E_\val^\sigma = E_\val \cup \{(v,\sigma(v))\mid v \in V\}$: this translates into keeping all transitions, for which the offset is \emph{not} $0$.
As a result, $\sigma'$ would change some of those for which the offset already is $0$, but agree with $\sigma$ for all vertices where this is currently not the case. This is a slightly surprising choice, since to progress one intuitively has to improve on the transitions whose offset is positive---which are then the ones one keeps.

\subsection{Games with Efficient Objective Improvement}\label{ssec:efficient}

In this subsection, we consider sufficient conditions for finding better strategies efficiently.
Note that we only have to consider cases where the termination condition (Line 6 of Algorithm~\ref{alg:oi_alg}) is not met.

The simplest condition for efficiently finding better strategies is the existence of local improvements.
In particular, it is easy to find, for a given solution $\val$, strategies $\sigma'$ for both players such that $f_{\sigma'}(\val) \leq f_{\sigma''}(\val)$ holds for all strategies $\sigma''$.
When there are local improvements, we can obtain a better strategy simply by applying them.
This leaves the case in which there are no local improvements, but where $\val$ also does not \emph{define} strategies for the two players.
We have seen that we can obtain a better strategy by only swapping edges, for which the inequations are sharp (Lemma~\ref{lem:stale}).

We now describe two conditions that, when simultaneously met, allow us to efficiently find better strategies:
that the games are \emph{sharp} and \emph{improving}.

\subsection*{Sharp games}
To do this efficiently, it helps if there are always $|V|$ inequations that are sharp for a solution $\val$ that minimises $f_\sigma(\val)$ for some $\sigma$.
Of course, if $\val$ is a solution returned by the simplex method, then there must be at least $|V|$ sharp inequations, as $\val$ is the solution defined by making $|V|$ inequations sharp (namely those that form the basis).
By requiring that there are exactly $|V|$ sharp inequations we require that an optimal solution defines a basis.
We call such a set of inequations $H$ and games that define them \emph{sharp} games.

\begin{defi}[Sharp Game]
\label{def:sharp}
A game $\mathcal G$ is called \emph{sharp} for a solution $\val$ if, and only if, $\val$ satisfies exactly $|V|$ inequations $H$ computed by $\mathsf{Inequations}(\mathcal{G})$ sharply.
It is called \emph{sharp} if no solution satisfies strictly more than $|V|$ inequations sharply.
\end{defi}

\subsection*{Improving games}
The second condition, which allows us to identify better strategies efficiently, is to assume that, for every strategy $\sigma$ for both players, if a solution $\val$ defined by a basis is not optimal for $f_\sigma$ under the constraints $H$, then there is a single basis change that improves it.
We call such games \emph{improving} and show that all sharp games are improving.

DPGs are not always improving. Being improving is a very useful property as it removes the problem that the simplex method can stall, as a single basis change is enough for improving games to improve the value of the objective function (hence the name).
For non-improving games, adjacent solutions might not improve although the current basis does not define a solution that is optimal for the current objective function.

\begin{obs}
Bases define solutions, and solutions identify (sets of) strategies: as we have used in Corollary \ref{cor:better}, it is easy to identify for a solution $\val$ the pair of strategies $\sigma$ (or the pairs of strategies if there are many) for whom $f_\sigma(\val)$ is the smallest among all pairs of strategies; they are simply those pairs that minimise the offset for every vertex.
\end{obs}

We therefore say that a solution identifies these pairs of strategies, and that a basis defines the pairs of strategies defined by the solution it defines.

\begin{defi}[Improving Game]
\label{def:impr}
A game $\mathcal G$ is called $\sigma$-\emph{improving} for an objective function $f_\sigma$ defined by a joined strategy $\sigma$ and a non-optimal solution $\val$ defined by a basis for $f_\sigma$ if, and only if, there is a single basis change that leads to a better solution $\val'$ (having a smaller value of $f_\sigma(\val')<f_\sigma(\val)$) defined by the new basis.

A game $\mathcal G$ is called \emph{improving} if it is improving for all objective functions $f_\sigma$ defined by a joined strategy $\sigma$ and all solutions $\val$ defined by bases that are not optimal w.r.t.\ $f_\sigma$.
\end{defi}

\begin{thm}\label{thm:sharpImprove}
Every sharp game $\mathcal G$ is improving.
\end{thm}

\begin{proof}
As our definition of sharp games simply says that every basic solution is nondegenerate, every step of a simplex algorithm (outside of optimal solutions) will provide a strict improvement \cite[Theorem 3.3]{DBLP:books/daglib/0014645}, so that $\mathcal G$ is $\sigma$-improving for every joined strategy $\sigma$ and all solutions $\val$ defined by bases that are not optimal w.r.t.\ $f_\sigma$.    
\end{proof}

We call a solution $\val'$ whose basis can be obtained from that of $\val$ by a single change to the basis of $\val$ a neighbouring solution to $\val$. 
We show that, for improving games, we can refine the result of Lemma~\ref{lem:stale} so that the better strategy $\sigma'$ also guarantees $f_{\sigma'}(\val')<f_\sigma(\val)$ for some neighbouring solution $\val'$ to $\val$. 

This allows us to consider $O(|E|)$ basis changes and, where they define a solution, to seek optimal strategies for a given solution.
Finding an optimal strategy for a given solution is straightforward.

\begin{thm}\label{theo:improve}
Let $\mathcal{G}= (V_{\min}, V_{\max}, E, w, \lambda)$ be an improving DPG, $\sigma$ a strategy for both players that is not co-optimal, $\val$ an optimal solution returned at Line 5 of Algorithm~\ref{alg:oi_alg} for $f_\sigma$. 
Then there is (at least) one neighbouring solution $\val''$ to $\val$ such that there is a better strategy $\sigma'$ that satisfies $f_{\sigma'}(\val'')<f_\sigma(\val)$.

This strategy $\sigma'$ can be selected in such a way that $(v,\sigma'(v)) \in E_\val^\sigma$ holds for all $v\in V$ for the given set $E_\val^\sigma$.
\end{thm}

\begin{proof}
By Lemma~\ref{lem:stale}, we can obtain a better strategy $\sigma'$ from transitions in $E_\val^\sigma$; for $\sigma'$ therefore $f_{\sigma'}(\val)=f_\sigma(\val)$ holds, and $\val$ is not optimal for $f_{\sigma'}$.

Let $\mathbf{b}$ be a basis that defines $\val$. 
As $\mathcal G$ is improving, it is $\sigma'$-improving, so that there is a basis $\mathbf{b}'$ that neighbours $\mathbf{b}$ and defines a solution $\val'$ with $f_{\sigma'}(\val')<f_{\sigma'}(\val)=f_\sigma(\val)$.
\end{proof}

We observe that any set $E^\sigma_\val$ can be selected, including the set $E_\val \cup \{(v,\sigma(v)\mid v \in V\}$, to select a better strategy $\sigma'$ from.

While we can restrict the selection of $\sigma'$ to the strategies that comply with the restriction $(v,\sigma'(v)) \in E_\val^\sigma$, there is no particular reason for doing so; as soon as we have a neighbouring solution $\val'$, we can identify a pair of strategies $\sigma'$ for which $f_{\sigma'}(\val')$ is minimal and select $\sigma'$ if $f_{\sigma'}(\val')< f_\sigma(\val)$ holds.

\subsection{Making Games Sharp (and Thus Improving)}\label{ssec:sharpeing}

Naturally, not every game is sharp.
In this subsection, we 
discuss how to almost surely make games sharp by adding sufficiently small random noise to the edge weights.
Note that these are `global' adjustments of the game that only need to be applied once, as it is the game that becomes sharp.

We first create notation for expressing how much we can change edge weights by adding small noise, such that joint co-optimal strategies of the resulting game are joint co-optimal strategies in the original game.
To this end, we define the \emph{gap} of a game.

\subsection*{Gap of a game}
Before defining the gap of the game, we recall that $\val_{\sigma}$ is the value of the joint strategy $\sigma$ (whereas $\val$ denotes the outcome of the optimisation problem).
We use this solution to argue that a small distortion of the edge weights cannot turn non-co-optimal strategies into co-optimal ones, and we define the gap of the game to reason about how small is small enough to retain non-co-optimality.
The value of a joint strategy itself is useful for this because it is much easier to access than the result of optimisation.
As a valuation of a non-co-optimal joint strategy is not a solution of a game, we will encounter negative offsets: they occur for exactly those edges, for which the associated inequation is not satisfied by $\val_{\sigma}$.

\begin{defi}[Contraction and Gap of a Game]
\label{def:gap}
Given a game $\mathcal{G}$, we call its \emph{contraction} $\lambda^*=\max\{\lambda_e\mid e \in E\}$.
For a non-co-optimal joint strategy $\sigma$, we call the \emph{gap} of $\sigma$ for $\mathcal{G}$ the value $\gamma_\sigma = \min \{-\off(\val_{\sigma},e) \mid e \in E,\off(\val_{\sigma},e)<0\}$.
We call the \emph{gap} of $\mathcal G$ the value $\gamma = \min\{\gamma_\sigma \mid \sigma$ is a non-co-optimal joint strategy$\}$, i.e.\ the minimal such value. 
\end{defi}

Note that $\gamma_\sigma >0$ always holds for non-co-optimal strategies; for the minimal%
\footnote{This skips over the case where all strategies are co-optimal, but that case is trivial to check and such games are trivial to solve, so that we ignore this case in this subsection.}
value $\gamma$, $\gamma >0$ thus always holds.

In a strategy improvement approach, where there is an edge $e$ with $-\off(\val_{\sigma},e)>0$, this would constitute a profitable switch, allowing the player to improve their return.
While we do not use it this way, we argue that the gap $\gamma$ allows us to infer that, when we change all edge weights by a sufficiently small factor (which we can derive from $\gamma_\sigma$, or $\gamma$, and the contraction of the game), a non-co-optimal strategy $\sigma$ is still non-co-optimal after adjusting the weights this way.

We now use the gap of a game $\gamma$ to define the magnitude of a change to all weights, such that all strategies that used to have a gap still have one, and thus that all co-optimal joint strategies from the distorted game are also co-optimal joint strategies in the undistorted game.

\begin{lem}\label{lem9}
Let $\mathcal G = (V_{\min}, V_{\max}, E, w, \lambda)$ be a DPG with contraction $\lambda^*$ and gap $\gamma$, and let $\mathcal G' = (V_{\min}, V_{\max}, E, w', \lambda)$ differ from $\mathcal G$ only in the edge weights such that, for all $e \in E$, $|w_e - w_e'| < \frac{1-\lambda^*}{2}\gamma$ holds.
Then any joint co-optimal strategy from $\mathcal G'$ is also co-optimal for $\mathcal G$.
\end{lem}
\begin{proof}
We argue that the small weight disturbance, $|w_e - w_e'| < \frac{1-\lambda^*}{2}\gamma$ for all $e \in E$, provides a small difference in the value that does not affect co-optimality. By using the definition of $\val_{\sigma}$ and the triangle inequality we can estimate the difference between the values. Precisely, for all joint strategies $\sigma$, we have for $\val_{\sigma}$ on $\mathcal G$, and $\val_{\sigma}'$ on $\mathcal G'$, that $|\val_{\sigma}(v) - \val_{\sigma}'(v)| < \frac{1}{1-\lambda^*} \frac{1-\lambda^*}{2}\gamma = \frac{\gamma}{2}$.
Indeed, $\sigma$ defines a play $\rho = e_0 e_1 e_2 \ldots$, and we have $\val_{\sigma}(v) = \out(\rho)=\sum_{i=0}^{\infty} w_{e_i} \prod_{j=0}^{i-1}\lambda_{e_j}$
and $\val_{\sigma}'(v)=\sum_{i=0}^{\infty} w_{e_i}' \prod_{j=0}^{i-1}\lambda_{e_j}$.
This provides 
$$|\val_{\sigma}(v) - \val_{\sigma}'(v)| \leq \sum_{i=0}^{\infty} |w_{e_i}-w_{e_i}'| \prod_{j=0}^{i-1}\lambda_{e_j} < \frac{1-\lambda^*}{2}\gamma \sum_{i=0}^{\infty}\prod_{j=0}^{i-1}\lambda_{e_j} \leq \frac{1-\lambda^*}{2}\gamma \sum_{i=0}^{\infty} (\lambda^*)^i = \frac{\gamma}{2}. $$

If we assume that $\sigma$ is not co-optimal for $\mathcal G$, then we have an edge $e = (v,v')$ with $-\off(\val_{\sigma},e)=\gamma_\sigma \geq \gamma$.
Using $\off'$ to indicate the use of $w_e'$ for $\mathcal G'$, by applying the triangle inequality we get

\[\begin{array}{cl}
|\off'(\val_{\sigma}',e)-\off(\val_{\sigma},e)| \!\!\!\!& \leq |\val_{\sigma}'(v)-\val_{\sigma}(v)|+|w_e - w_e'| + \lambda_e|\val_{\sigma}(v')-\val_{\sigma}'(v')|  \\
& < \frac{\gamma}{2} + \frac{1-\lambda^*}{2}\gamma + \lambda_e \frac{\gamma}{2} \leq \gamma . \\
\end{array}\]

Together with the fact that $-\off(\val_{\sigma},e) \geq \gamma$, this provides $\off'(\val_{\sigma}',e)<0$.

Thus, $\sigma$ is not co-optimal for $\mathcal G'$ either.
\end{proof}

\begin{lem}\label{l10}
Given a DPGs $\mathcal{G}= (V_{\min}, V_{\max}, E, w, \lambda)$, the DPG $\mathcal G'= (V_{\min}, V_{\max}, E, w', \lambda)$ resulting from $\mathcal G$ by adding independently uniformly at random drawn values from an interval $(-\varepsilon,\varepsilon)$ to every edge weight, will almost surely result in a sharp game.
\end{lem}

Note that these distributions are continuous.
They could be replaced by any other distribution over these intervals that has weight $0$ for all individual points.

\begin{proof}
In this proof, we show a stronger property: every two different bases almost surely define different valuations (regardless of whether or not these valuations are solutions).

To this end, we can select two arbitrary different sets of $|V|$ edges that define potential bases $\mathbf b_1$ and $\mathbf b_2$; they can be selected before drawing $\mathcal G'$.

What we show is that it happens with probability $0$ that $\mathbf b_1$ and $\mathbf b_2$ are bases \emph{and} define the same valuations (regardless of whether or not these valuations are solutions) in $\mathcal G$.
If $\mathbf b_1$ or $\mathbf b_2$ is not a basis (which would happen if the inequations are not linearly independent, e.g.\ if a vertex $v$ does not appear at all in these inequations) we are done, too.

As $\mathbf b_1$ and $\mathbf b_2$ are different, there will be one inequation that corresponds to an edge $e = (v,v')$ that occurs in $\mathbf b_2$, but not in $\mathbf b_1$.
As all weight disturbances are drawn independently, we assume without loss of generality that the weight disturbance to this edge is drawn last.

Now, the valuation $\val_1$ defined by $\mathbf b_1$ does not depend on this final draw.
For $\val_1$, there is a value $w_e' = \val_1(v)-\lambda_e\val_1(v')$ that defines the weight $w_e'$ that $e$ would need to have, in order to make the inequation sharp for $e$.

For the valuation $\val_1$ defined by $\mathbf b_1$ to be equal to the valuation $\val_2$ defined by $\mathbf b_2$, the weight for the edge $e$ (after adding the drawn distortion) needs to be exactly $w_e'$. (This is a necessary condition, but not necessarily a sufficient condition.)
There is at most one value for the disturbance that would provide for this, and this disturbance for the weight of $e$ is sampled with a likelihood of $0$.
\end{proof}

Putting these two results together, we get the following:

\begin{cor}
Given a pair of DPGs $\mathcal G = (V_{\min}, V_{\max}, E, w, \lambda)$ with contraction $\lambda^*$ and gap $\gamma$, and $\mathcal G' = (V_{\min}, V_{\max}, E, w', \lambda)$ obtained from $\mathcal G$ by adding independently uniformly at random drawn values from an interval $(-\varepsilon,\varepsilon)$ to every edge weight, for some $\varepsilon < \frac{1-\lambda^*}{2}\gamma$, it holds that any joint co-optimal strategy from $\mathcal G'$ is also co-optimal for $\mathcal G$, and $\mathcal G'$ is almost surely sharp.
\end{cor}

Note that we can estimate the gap cheaply when all coefficients in $\mathcal G$ are rational.
To estimate the gap
$$\gamma = \min \{-\off(\val_{\sigma},e) \mid \sigma \mbox{ is a non-co-optimal joint strategy},e \in E,\off(\val_{\sigma},e)<0\}\ ,$$
we first fix a joint non-co-optimal strategy $\sigma$ and an edge $e=(v,w)$ with $\off(\val_{\sigma},e)<0$.

Starting in $v$, $\sigma$ defines a run $\rho =  e_1 e_2 e_3 \ldots$ in the form of a ``lasso path", which consists of a (possibly empty) initial path $e_1,\ldots,e_k$, followed by an infinitely often repeated cycle $e_1',\ldots,e_\ell'$, where all edges occur only once.

In this run, an edge $e_i$ occurs only once and contributes $(\prod_{j=1}^{i-1}\lambda_{e_j}) w_{e_i}$ to the value of this run, while an edge $e_i'$ occurs infinitely often and contributes the value $\frac{(\prod_{j=1}^{k}\lambda_{e_j}) \cdot (\prod_{j=1}^{i-1}\lambda_{e_j'})\cdot w_{e_i'}}{1-\prod_{j=1}^{\ell}\lambda_{e_j'}}$ to the value of the run.
Now all we need to do is to estimate a common denominator.

To do this, let $\mathsf{nom}(r)$ and $\denom(r)$ be the nominator and denominator of a rational number $r$.
It is easy to see that 
$$\mathsf{common'} = \prod\limits_{v\in V}\denom(\lambda_{(v,\sigma(v))}) \cdot \denom(w_{(v,\sigma(v))})$$
is a common denominator of all expressions of the first type, ($\prod_{j=1}^{i-1}\lambda_{e_j})w_{e_i}$, and of the nominator from the second, $(\prod_{j=1}^{k}\lambda_{e_j}) \cdot (\prod_{j=1}^{i-1}\lambda_{e_j'})\cdot w_{e_i'}$.

This leaves the contribution of the denominator of the fraction from the second form, which is the nominator of the term $1-\prod_{j=1}^{\ell}\lambda_{e_j'}$.
But since the term is between $0$ and $1$, its value is strictly smaller than the value of its denominator; we thus have for
$$n_1 =\mathsf{nom}(1-\prod_{j=1}^{\ell}\lambda_{e_j'})
< \prod_{j=1}^\ell \denom(\lambda_{e_j'})\leq \prod\limits_{v\in V}\denom(\lambda_{(v,\sigma(v))})$$
that
$$\mathsf{common'}\cdot n_1
\mbox{ is a denominator for }
\out(\rho)\ .$$

To obtain $-\off(\val_{\sigma},e)$, we have to compare with a run starting in $v$, but taking the edge $e \neq (v,\sigma(v))$ first.

This defines a run $\rho' = d_1 d_2 d_3 \ldots$ (with $d_1 = e$), again in the form of a ``lasso path" that consists of an initial path $d_1,\ldots,d_{k'}$, followed by an infinitely often repeated cycle $d_1',\ldots,d_{\ell'}'$, where every edge on the path $d_1,\ldots,d_{k'},d_1',\ldots,d_{\ell'}'$ occurs only once.

In this run, an edge $d_i$ occurs only once and contributes $(\prod_{j=1}^{i-1}\lambda_{d_j}) w_{d_i}$ to the value of this run, while an edge $d_i'$ occurs infinitely often and contributes the value $\frac{(\prod_{j=1}^{k'}\lambda_{d_j}) \cdot (\prod_{j=1}^{i-1}\lambda_{d_j'})\cdot w_{d_i'}}{1-\prod_{j=1}^{{\ell'}}\lambda_{d_j'}}$ to the value of the run.

It is easy to see that 
$$\mathsf{common'} \cdot \denom(\lambda_e) \cdot \denom(w_e)$$
is a common denominator of all expressions of the first type, ($\prod_{j=1}^{i-1}\lambda_{d_j})w_{d_i}$, and of the nominator from the second, $(\prod_{j=1}^{k'}\lambda_{d_j}) \cdot (\prod_{j=1}^{i-1}\lambda_{d_j'})\cdot w_{d_i'}$.

This again leaves the contribution of the denominator of the fraction from the second form, which is the nominator of the term $1-\prod_{j=1}^{{\ell'}}\lambda_{d_j'}$.
Again a term between $0$ and $1$, its value is strictly smaller than the value of its denominator; we thus have for
$$n_2 = \mathsf{nom}(1-\prod_{j=1}^{{\ell'}}\lambda_{d_j'})
< \prod_{j=1}^{\ell'} \denom(\lambda_{d_j'})\leq \prod\limits_{v\in V}\denom(\lambda_{(v,\sigma(v))})$$
that
$$\mathsf{common'}\cdot \denom(\lambda_e) \cdot \denom(w_e) \cdot n_2
\mbox{ is a denominator for }
\out(\rho')\ .$$
Consequently,
$$\mathsf{common'}\cdot \denom(\lambda_e) \cdot \denom(w_e) \cdot n_1 \cdot n_2
\mbox{ is a denominator for }
\off(\val_{\sigma},e)\ .$$

Putting the estimates together, we get that
$$\mathsf{bound} = \denom(\lambda_e) \cdot \denom(w_e) \cdot \prod\limits_{v\in V}\denom(\lambda_{(v,\sigma(v))})^3 \cdot \denom(w_{(v,\sigma(v))})$$
is smaller than the smallest denominator of $\off(\val_{\sigma},e)$.
We then use this to estimate $\gamma \geq -\off(\val_{\sigma},e) > 1/\mathsf{bound}$.

This value is easily estimated by using the highest possible denominators available in $\mathcal G$---maximising $\denom(\lambda_e) \cdot \denom(w_e)$ over all edges $e \in E$ and maximising $\denom(\lambda_{(v,\sigma(v))})^3 \cdot \denom(w_{(v,\sigma(v))})$ individually for every vertex $v \in V$.
Its representation is polynomial in the size of $\mathcal G$.

Thus, we can almost surely obtain sharpness by adding small noise to the weights, and the resulting sharp games are also improving by Theorem \ref{thm:sharpImprove}.
This guarantees cheap progress for the case where there are no local improvements.

\subsection{Mixing Pivoting on the Simplex and of the Objective}
\label{ssec:mixing}

When using a simplex based technique to implement $\mathsf{LinearProgramming}$ (Line 5 of Algorithm~\ref{alg:oi_alg}), then the algorithm mixes three
approaches that stepwise reduce the value of $f_\sigma(\val)$:
\begin{enumerate}
    \item The simplex algorithm updates the basis, changing $\val$ (while retaining the objective function $f_\sigma$).
    \item Local updates, that change the objective function $f_\sigma$ (through updating $\sigma$)
    and retain $\val$.
    \item Non-local updates.
\end{enumerate}

Non-local updates are more complex than the other two, and the correctness proofs make use of the optimality (w.r.t.\ $f_\sigma$) of the current solution. For both reasons, it seems natural to take non-local updates as a last resort.

The other two updates, however, can be freely mixed, as they both bring down the value of $f_\sigma(\val)$ by applying local changes.
That the improvements from (1) are given preference in the algorithm is a choice made to keep the implementation of the algorithm for using linear programs open, allowing, for example, to use ellipsoid methods \cite{Kha79} or inner point methods \cite{Kar84} to keep this step tractable.

\section{Experimental Evaluation}\label{sec:exp}

To evaluate the performance of the novel approach, we have implemented it in an experimental framework in C++. Together with the objective improvement algorithm (OI, for short), we have also implemented the classic asymmetric variant of the strategy improvement algorithm (SI, for short) that has been described in Section~\ref{sec:outline}.
Note that our implementation uses floating-point numbers with double precision (type \texttt{double}), while the approach itself assumes infinite precision. As a consequence, match checks are not possible and the solutions can be affected by an $\epsilon$ error due to the precision of the library (that comes from the external linear programming solver too).
For most parts, this could be rectified by using fractions. However, while fractions with infinite precision are available with the right libraries, their use would lead to an overhead and would not solve the precision problem of the external linear programming solver.
We also note that the small noise added in Section \ref{ssec:sharpeing} to almost surely make a game sharp and improving requires a continuous value space.

Evaluation comes with a number of difficulties. First, there are no actual concrete benchmarks for DPGs.
We have therefore translated parity games that model synthesis problems into DPGs. The translation process is exponential in the number of priorities of the parity games. For this reason we could translate and test only two types of concrete games: the Elevator and Language Inclusion. 
Besides these concrete games, we have also generated random games of various sizes.
We have considered games with rational weights in the interval $[-250,250]$ on the edges, as games with integer weights represent a restricted class of games.
The range of rational weights is not crucial as it suffice to multiply their values to scale the interval up or down.
Moreover, we use a uniform discount factor of $0.9$, since for games in which each vertex has different weights on the edges, non-uniform discount factors do not make games any easier or harder to solve.
Using a low discount factor (\eg\ 0.5 or below) make, instead, games easier to solve, because with small discount factors the solution (\ie\ the valuation $\val_{\mathcal G}$) of the vertices converges after few steps, while with a discount factor close to 1 the valuation requires many steps to converge.
Finally, it is worth noting that games with such random weights on the edges are also almost surely sharp, hence, it has not been necessarily to implement the \emph{sharpening} technique described in Section~\ref{ssec:sharpeing}.

The second obstacle is that, for an implementation, one has to choose \emph{how} to update the current ``state", both for SI and OI.
For SI, we use a greedy all switch rule: we update all local decisions where there is a local improvement available, and where there are multiple, we choose one with the maximal local improvement.
For OI, there are more choices to make.
As mentioned in Section \ref{ssec:mixing}, in simplex-based approaches, we can freely mix (1) the standard updates of the basis for a fixed objective function, and (2) updates of the objective function itself, and in most steps, both options are available.
For (1), there are a plethora of update rules for the simplex, and for (2), there is also a range of ways of how to select the update.
We have dodged this question by not directly implementing a solver, but by, instead, solving the linear programs with an off-the-shelf solver. We have used the LP solver provided by the ALGLIB library\footnote{Library libalglib version 3.16.0} for both OI and SI.
This does not restrict us to the use of simplex-based solvers, while it would restrict simplex-based solutions to always giving priority to (1) over (2).
Where the solver does provide a solution for a given objective function, we use rule (2) with a greedy all switch rule.

One cost of doing this is that we do not have a similarly straightforward integration of (3) non-local updates, where we are in the optimal solution for an objective function that cannot be improved with rule (2).
Instead of using the method from Theorem \ref{theo:improve} to slightly adjust a simplex method to find an improvement for improving games (knowing that the games we create are almost surely improving), we use the general rule from Lemma \ref{lem:stale}, trying random stale strategies.
In principle, this might lead to large plateaux that are hard to exit, but we considered this very unlikely, at least for the improving games we create, and our experiments have confirmed this.

This setup allows us to fairly compare OI with SI by counting the number of linear program instances that need to be solved on the way.

\begin{figure}
    \begin{tikzpicture}[scale=1.2,every node/.style={scale=0.85}]
      \begin{axis}[
        xmin=100, ymin=0,
        ymajorgrids=true,
        xmajorgrids=true,
        ytick={0, 5, 10, 15, 20, 25, 30, 35, 40},
        xtick={100, 200, 300, 400, 500, 600, 700, 800, 900, 1000},
        enlarge x limits=false,
        enlarge y limits=upper,
        mark size=1.2pt,
        ylabel={Number of iterations (avg)},
        xlabel={Game size},
        yticklabel style = {font=\tiny,xshift=-0.5ex},
        xticklabel style = {font=\tiny,yshift=-0.5ex},
        legend style={nodes={scale=0.6, transform shape}},
        legend pos=north west
        ]
        \addplot[color=red, mark=o] table [x=RUN, y=DSI_i] {results2.txt};
        \addlegendentry{SI}
    \addplot[color=blue, mark=star, mark size=1.4] table [x=RUN, y=OI_i] {results2.txt};
        \addlegendentry{OI}
        \end{axis}
    \end{tikzpicture}
    \centering
    \caption{Comparisons of the average number of iterations (LP calls) on games with two successors for each vertex.}
    \label{fig:restwoA}
\end{figure}
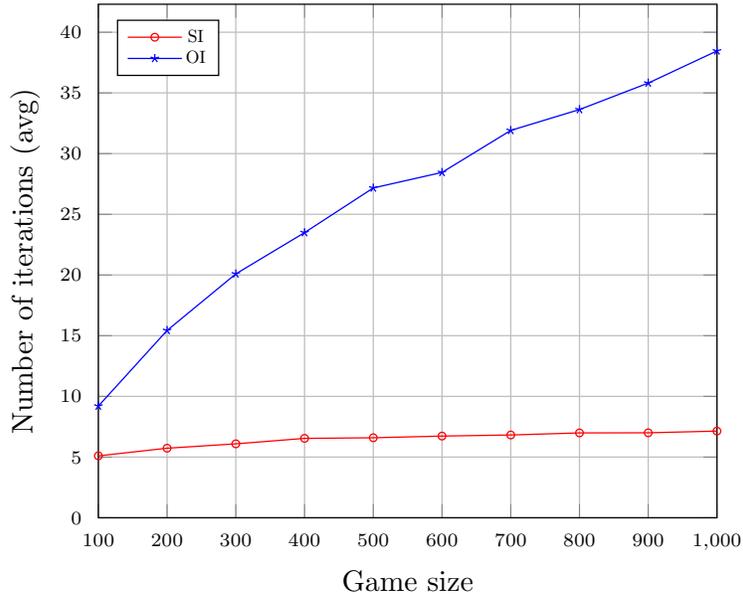

As our first benchmark set we selected games with two outgoing edges per vertex.
This type of games have a relatively small strategy space, which should make them a relatively easy target for SI in that SI should require fewer steps to identify the optimal strategy than OI.
Figures~\ref{fig:restwoA} and~\ref{fig:restwoB} show the number of iterations and strategy updates required to solve a game of size ranging from 100 to 1000 vertices. We selected 1000 games, each of the ten points in the graph shows the mean value for a cluster of 100 games of the corresponding size.

The number of iterations (Figure~\ref{fig:restwoA}) represents how many times the linear programming solver has been called.
As expected, SI is more efficient than OI at solving these games in terms of LP calls.
What was less expected is how the advantage of SI over OI grows with the size of the game. Solving the linear programs, instead, is simpler for OI, at least when using a simplex style algorithm%
\footnote{Simplex style algorithms 
operate in two phases: in the first phase, they look for any basis that defines a solution to the constraint system (Line 1), and in the second, they optimise (while loop).
For OI, the first phase can be skipped for all but the first call; this is because the constraint system does not change, so that the previous solution can be used as a starting point.
For SI, the constraint system changes; here, the previous solution is never a solution to the new linear program.}.
We also note that only 55 (5.5\%) of the 1000 games considered overall required a non-local improvement step.
In these rare cases, this has been solved by choosing a stale switch (i.e.\ a switch of strategy that leads to the same solution).

The number of updates, instead, counts how many times each vertex changed its strategy before the next call to the LP solver, and then provides the sum of these switches. 
If the ownership of the vertices of a game is not balanced (\eg\ one player owns 80\% of the vertices), most of the solving effort is charged to the LP solver, while the SI algorithm itself needs to update very few strategies.
On the contrary, OI updates the strategies for both the players.
It would therefore stand to reason to expect that the number of local strategy updates would give an extra advantage to SI.
This is, however, not the case: Figure \ref{fig:restwoB} shows that, while OI needs to update the strategy for at least twice as many vertices, it needs only around a third more local updates for this.

\begin{figure}
    \begin{tikzpicture}[scale=1.2,every node/.style={scale=0.85}]
      \begin{axis}[
        xmin=100, ymin=0,
        ymajorgrids=true,
        xmajorgrids=true,
        ytick={45, 100, 200, 300, 400, 500, 600, 700},
        xtick={100, 200, 300, 400, 500, 600, 700, 800, 900, 1000},
        y tick label style={/pgf/number format/fixed},
        enlarge x limits=false,
        enlarge y limits=upper,
        mark size=1.2pt,
        ylabel={Number of updates (avg)},
        xlabel={Game size},
        yticklabel style = {font=\tiny,yshift=-0.5ex},
        xticklabel style = {font=\tiny,yshift=-0.5ex},
        legend style={nodes={scale=0.6, transform shape}},
        legend pos=north west
        ]
        \addplot[color=red, mark=o] table [x=RUN, y=DSI_s] {results2.txt};
        \addlegendentry{SI}
    \addplot[color=blue, mark=star, mark size=1.4] table [x=RUN, y=OI_s] {results2.txt};
        \addlegendentry{OI}
        \end{axis}
    \end{tikzpicture}
    \caption{Comparisons of the average number of local strategy updates for games with two successors per vertex.}
    \label{fig:restwoB}
\end{figure}
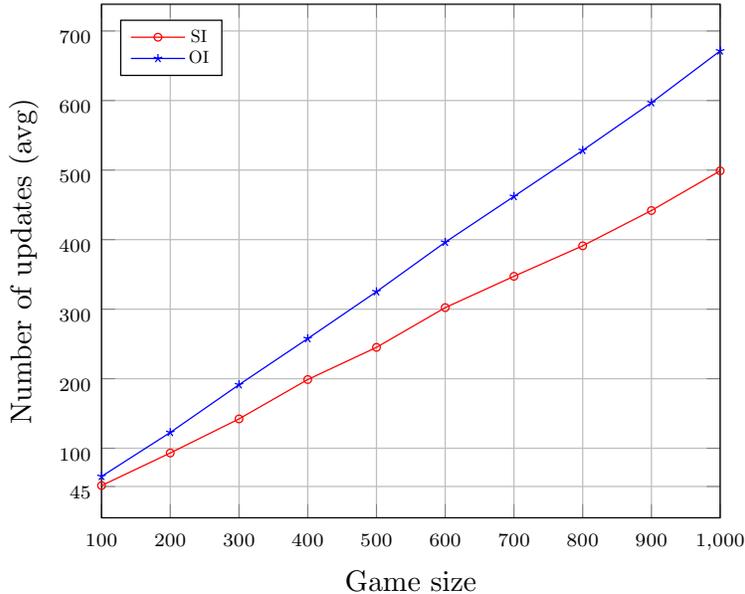

In the second benchmark set we selected games with more moves per vertex, in the range $[5,10]$. In this case, the space of strategies is much wider than one of the games of the first benchmark set, and algorithms based on SI find these games harder to solve. 
Figures~\ref{fig:resfiveA} and~\ref{fig:resfiveB} show the comparison of OI and SI on the second benchmark set of 1000 games.

In Figure~\ref{fig:resfiveA} we can observe a slowly but consistently growing gap in the number of LP calls (iterations). In addition to the advantage of simpler LP problems, we see that in these slightly more complex games OI also requires fewer calls, with SI needing some 2.5 to 3 times as many.
For this benchmark set, only 7 (0.7\%) of the games reached a non-improving state at any point.

Interestingly, Figure~\ref{fig:resfiveB}, instead, depicts a symmetric behaviour of the graph of Figure~\ref{fig:restwoB}, but this time with the lanes of the two solvers swapped. Thus, although OI considers the strategy at all vertices (instead of at most half), the sum of local improvements is also slightly smaller.

\begin{figure}
    \begin{tikzpicture}[scale=1.2,every node/.style={scale=0.85}]
      \begin{axis}[
        xmin=100, ymin=0,
        ymajorgrids=true,
        xmajorgrids=true,
        ytick={0, 1, 2, 3, 4, 5, 6, 7, 8, 9, 10, 11},
        xtick={100, 200, 300, 400, 500, 600, 700, 800, 900, 1000},
        enlarge x limits=false,
        enlarge y limits=upper,
        mark size=1.2pt,
        ylabel={Number of iterations (avg)},
        xlabel={Game size},
        yticklabel style = {font=\tiny,xshift=-0.5ex},
        xticklabel style = {font=\tiny,yshift=-0.5ex},
        legend style={nodes={scale=0.6, transform shape}},
        legend pos=north west
        ]
        \addplot[color=red, mark=o] table [x=RUN, y=DSI_i] {results5.txt};
        \addlegendentry{SI}
    \addplot[color=blue, mark=star, mark size=1.4] table [x=RUN, y=OI_i] {results5.txt};
        \addlegendentry{OI}
        \end{axis}
    \end{tikzpicture}
    \caption{Comparisons of the average number of iterations (LP calls) for games with many (5 to 10) successors per vertex.}
\label{fig:resfiveA}
\end{figure}
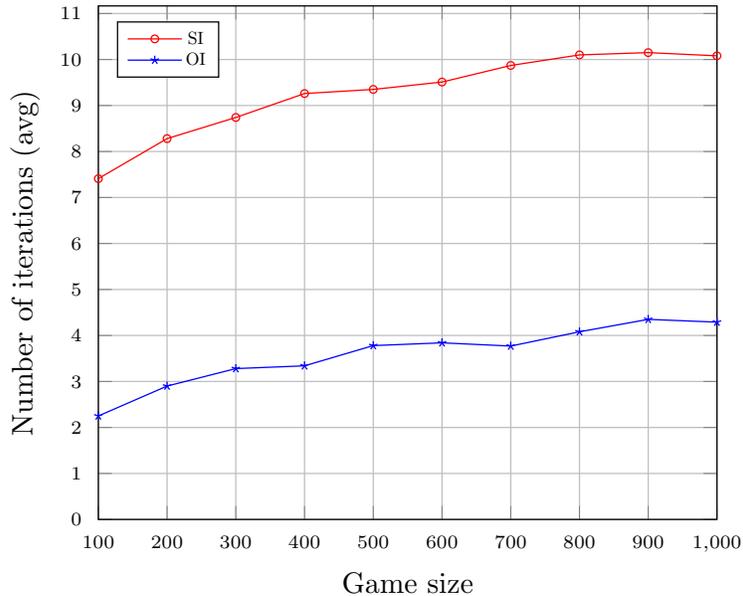

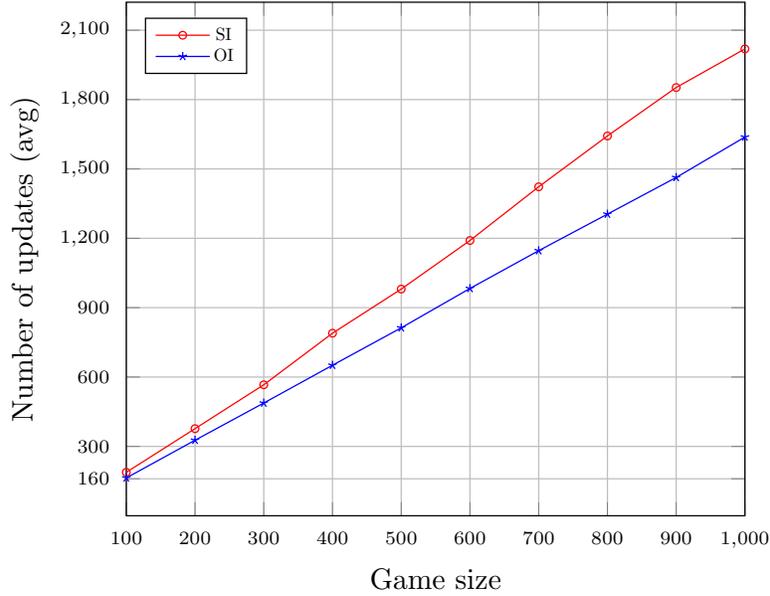
\begin{figure}
    \begin{tikzpicture}[scale=1.2,every node/.style={scale=0.85}]
      \begin{axis}[
        xmin=100, ymin=0,
        ymajorgrids=true,
        xmajorgrids=true,
        ytick={160, 300, 600, 900, 1200, 1500, 1800, 2100},
        xtick={100, 200, 300, 400, 500, 600, 700, 800, 900, 1000},
        y tick label style={/pgf/number format/fixed},
        enlarge x limits=false,
        enlarge y limits=upper,
        mark size=1.2pt,
        ylabel={Number of updates (avg)},
        xlabel={Game size},
        yticklabel style = {font=\tiny,xshift=-0.5ex},
        xticklabel style = {font=\tiny,yshift=-0.5ex},
        legend style={nodes={scale=0.6, transform shape}},
        legend pos=north west
        ]
        \addplot[color=red, mark=o] table [x=RUN, y=DSI_s] {results5.txt};
        \addlegendentry{SI}
    \addplot[color=blue, mark=star, mark size=1.4] table [x=RUN, y=OI_s] {results5.txt};
        \addlegendentry{OI}
        \end{axis}
    \end{tikzpicture}
    \caption{Comparisons of the average number of local strategy updates for games with many (5 to 10) successors per vertex.}
\label{fig:resfiveB}
\end{figure}

When considering games with many transitions, namely games where the outdegree of every vertex is 10\% of the number of vertices, the advantage of OI over SI grows further.
Figure~\ref{fig:resmoreA} shows a larger gap in the number of LP calls than Figure~\ref{fig:resfiveA}, and
Figure~\ref{fig:resmoreB} shows not only an advantage of OI over SI, but it also appears that the number of local strategy updates grows linearly for OI, but faster for SI. 
Only 3 (0.3\%) of the games reached a non-improving state at some point of the solving process.

Finally, for the concrete problems translated from parity games, all these games turned out to be easy to solve so that a single LP call suffices to find the optimal solution. Therefore, in Table~\ref{tab:prty}, instead of showing the number of iterations and updates, we show the solution time in seconds that provides a measure of the complexity of the game in terms of size of the linear system. Most of the games can be solved in less than one second. For the Elevator family of games, their size grows rapidly, so that we could not translate more than the first 5 instances. For the Language Inclusion games, instead, we report the games solved within 15 minutes.

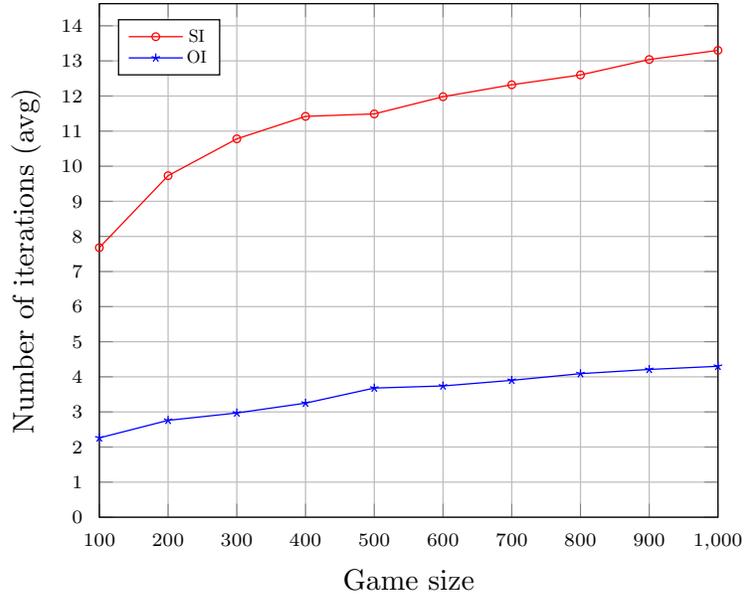
\begin{figure}
    \begin{tikzpicture}[scale=1.2,every node/.style={scale=0.85}]
      \begin{axis}[
        xmin=100, ymin=0,
        ymajorgrids=true,
        xmajorgrids=true,
        ytick={0, 1, 2, 3, 4, 5, 6, 7, 8, 9, 10, 11, 12, 13, 14},
        xtick={100, 200, 300, 400, 500, 600, 700, 800, 900, 1000},
        enlarge x limits=false,
        enlarge y limits=upper,
        mark size=1.2pt,
        ylabel={Number of iterations (avg)},
        xlabel={Game size},
        yticklabel style = {font=\tiny,xshift=-0.5ex},
        xticklabel style = {font=\tiny,yshift=-0.5ex},
        legend style={nodes={scale=0.6, transform shape}},
        legend pos=north west
        ]
        \addplot[color=red, mark=o] table [x=RUN, y=DSI_i] {resultsN.txt};
        \addlegendentry{SI}
    \addplot[color=blue, mark=star, mark size=1.4] table [x=RUN, y=OI_i] {resultsN.txt};
        \addlegendentry{OI}
        \end{axis}
    \end{tikzpicture}
    \caption{Comparisons of the average number of iterations (LP calls) for games with a linear number of successors per vertex (10\%).}
\label{fig:resmoreA}
\end{figure}

\begin{figure}
    \begin{tikzpicture}[scale=1.2,every node/.style={scale=0.85}]
      \begin{axis}[
        xmin=100, ymin=0,
        ymajorgrids=true,
        xmajorgrids=true,
        ytick={160, 500, 1000, 1500, 2000, 2500, 3000, 3500, 4000, 4500, 5000, 5500, 6000, 6500, 7000},
        xtick={100, 200, 300, 400, 500, 600, 700, 800, 900, 1000},
        y tick label style={/pgf/number format/fixed},
        enlarge x limits=false,
        enlarge y limits=upper,
        mark size=1.2pt,
        ylabel={Number of updates (avg)},
        xlabel={Game size},
        yticklabel style = {font=\tiny,xshift=-0.5ex},
        xticklabel style = {font=\tiny,yshift=-0.5ex},
        legend style={nodes={scale=0.6, transform shape}},
        legend pos=north west
        ]
        \addplot[color=red, mark=o] table [x=RUN, y=DSI_s] {resultsN.txt};
        \addlegendentry{SI}
    \addplot[color=blue, mark=star, mark size=1.4] table [x=RUN, y=OI_s] {resultsN.txt};
        \addlegendentry{OI}
        \end{axis}
    \end{tikzpicture}
    \caption{Comparisons of the average number of local strategy updates for games with a linear number of successors per vertex (10\%).}
\label{fig:resmoreB}
\end{figure}
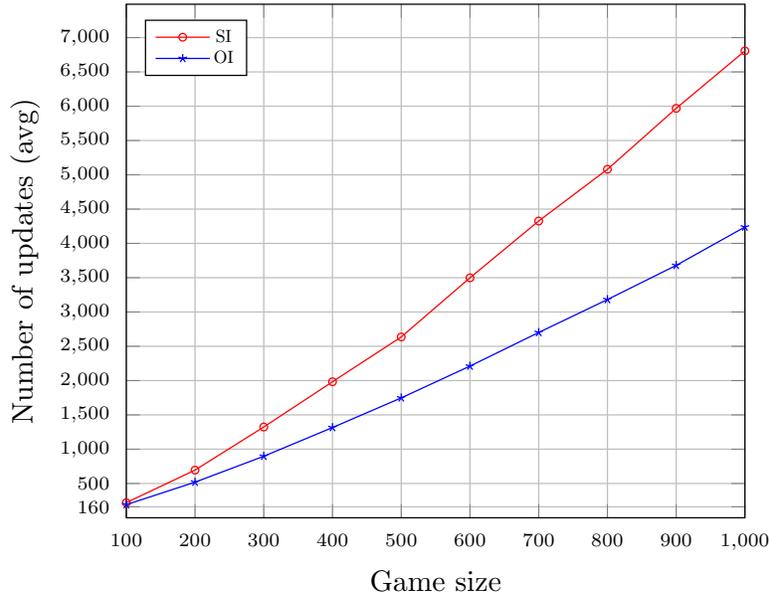

\begin{table}[htbp]
  \begin{center}
    {
      \begin{tabular}{|l|r|r||r|}
        \hline & & &\\[-.7em]
        Benchmark & Positions & Moves & Time
        \\ \hline
        \hline  & & &\\[-.7em]
        Elevator 1  & $36$  & $54$   & $0$
        \\
        Elevator 2  & $144$   & $234$   & $0$
        \\
        Elevator 3  & $564$ & $950$  & $0$
        \\
        Elevator 4  & $2688$ & $4544$  & $2$
        \\
        Elevator 5  & $15683$ & $26354$  & $16$
        \\ \hline & & &\\[-.7em]
        Language Inclusion 1  & $170$  & $1094$   & $0$
        \\
        Language Inclusion 2  & $304$  & $1222$   & $0$
        \\
        Language Inclusion 3  & $428$  & $878$   & $0$
        \\
        Language Inclusion 4  & $628$  & $1538$   & $0$
        \\
        Language Inclusion 5  & $509$  & $2126$   & $0$
        \\
        Language Inclusion 6  & $835$  & $2914$   & $0$
        \\
        Language Inclusion 7  & $1658$ & $4544$   & $1$
        \\
        Language Inclusion 8  & $14578$ & $17278$   & $47$
        \\
        Language Inclusion 9  & $25838$ & $29438$   & $279$
        \\
        \hline
      \end{tabular}
    }
    \captionsetup{type = table, justification = centering}
    \caption{\label{tab:prty} \small Experiments on concrete verification
      problems.}
  \end{center}
  \vspace{-1em}
\end{table}

\section{Discussion}\label{sec:discuss}

There is widespread belief that mean payoff and discounted payoff games have two types of algorithmic solutions: value iteration~\cite{FGO20,Koz21} and strategy improvement~\cite{Lud95,Pur95,BV07,Sch08a,STV15}.
We have added a third method, which is structurally different and opens a new class of algorithms to attack these games.
Moreover, our new symmetric approach has the same proximity to linear programming as strategy improvement algorithms, which is an indicator of efficiency.

The lack of a benchmarking framework for existing algorithms prevents us from testing and comparing an eventual implementation thoroughly, but we have compared OI against SI on random games and synthesis parity translated games in Section \ref{sec:exp}.
To keep this comparison fair (and simple), we have put as much as possible of both problems into LP calls, using greedy all switch updates in both cases.
We found that SI performs better on random games with very few (two) successors per vertex, while OI already comes out on top with few (five to ten), and shines with many (10\% of the vertices) successors.
This was a bit unexpected to us, as we thought that such a naive implementation of the new concept could not possibly compete.

Our results for random games with a low (five to ten) outdegree suggests that the number of local updates seen grows linearly with the size of the game for both SI and OI, and that the number of iterations initially grows and later plateaus.
This may well be the same for games with a minimal outdegree of two, and the data clearly supports this for SI, but the number of LP calls growth almost linearly for OI. This might be an outlier in the behaviour and it is well possible that OI plateaus later, as it is unlikely that games with a low and very low outdegree behave fundamentally different.
For SI, this was entirely expected, and it is perhaps not overly surprising for OI either.
Generally speaking, OI appears to do better than SI measured in the number of LP calls on random games, except where the outdegree is tiny.

Naturally, a fresh approach opens the door to much follow-up research.
A first target for such research is the questions on how to arrange the selection of better strategies to obtain fewer updates, either proven on benchmarks, or theoretically in worst, average, or smoothed analysis.
Our implementation effectively updates the objective function only once an optimal solution for the current objective function is found.
This is for two pragmatic reasons:
it allows us to use a solver for linear programs as a black box, and we feel that it provides the best comparison with strategy improvement.
In terms of the three update rules stated in Section \ref{ssec:mixing}, this approach used (1) basis change where available, then (2) update of the objective function through updating the joint strategy, before turning to (3) non-local updates as a last resort.

Seeing that (1) and (2) are equally cheap and simple, we could also favour (2) over (1):
where we have a solution $\val$ defined by a basis, we can first find a joint strategy $\sigma$ such that $f_\sigma(\val)$ is minimal among all strategies, before turning to (1) and (3).
This would also make sure that no corner of the polytope is ever visited twice. This could happen in principle when favouring (1) over (2), as a step that is good for one objective function might not be good for its successor, but as the value of the objective function goes down in every step, a return to a previously optimal solution is not possible.

From a theoretical perspective, it would in particular be interesting to establish non-trivial upper or lower bounds for various pivoting rules. Without such a study, a trivial bound for the proposed approach is provided by the number of strategies (exponential).

A second question is whether this method as whole can be turned into an inner point method \cite{Kar84}.
If so, this could be a first step towards showing tractability of discounted payoff games -- which would immediately extend to mean-payoff and parity games.

\subsection*{Acknowledgements}

\includegraphics[height=8pt]{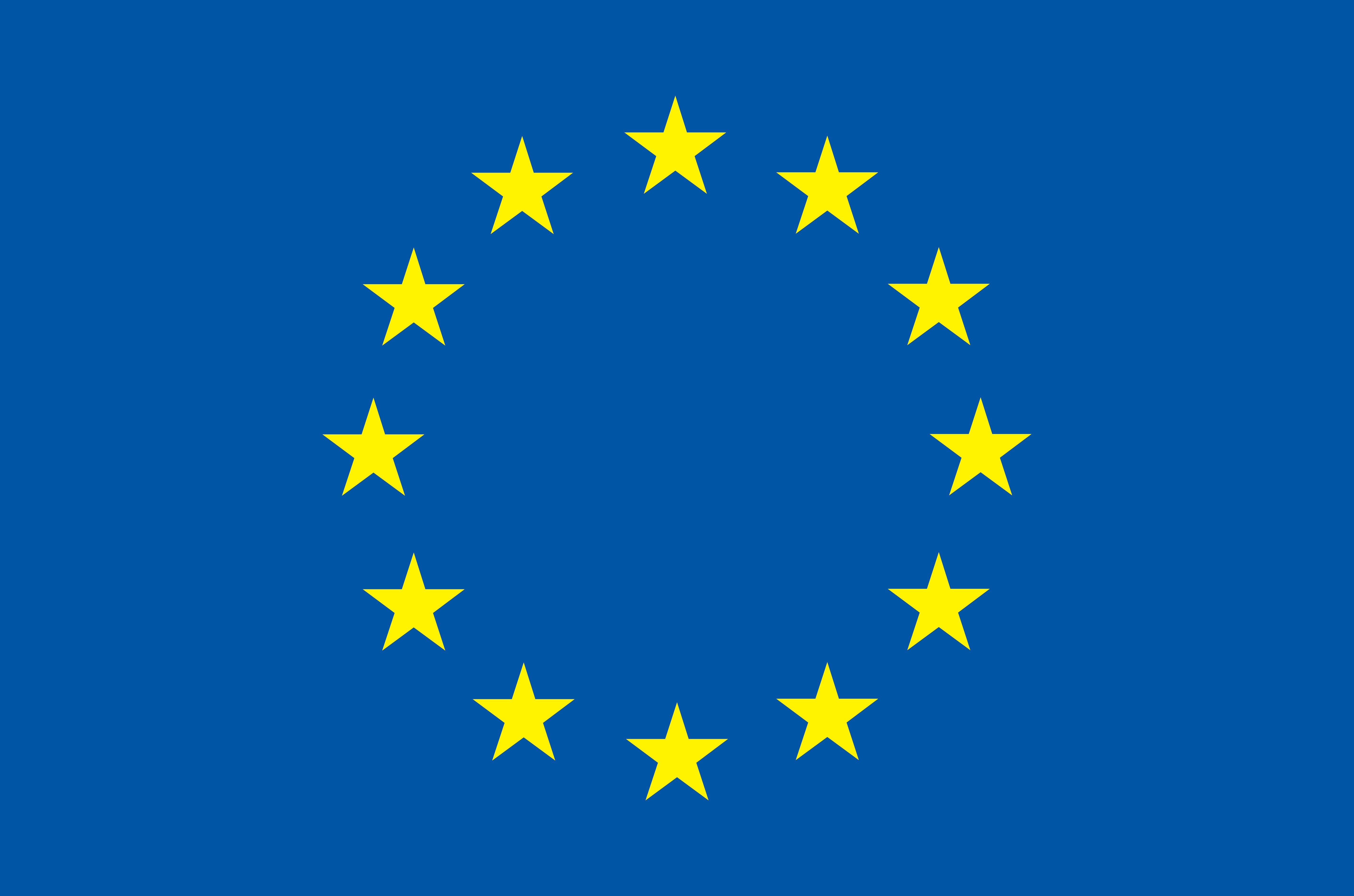} This project has received funding from the European Union’s Horizon 2020 research and innovation programme under the Marie Sk\l odowska-Curie grant agreement No 101032464.
It was supported by the EPSRC through the projects EP/X017796/1 (Below the Branches of Universal Trees) and
EP/X03688X/1 (TRUSTED: SecuriTy SummaRies for SecUre SofTwarE Development). We also thank Peter Austin from the University of Liverpool for providing an integration of the LP solver.
We thank an anonymous reviewer for pointing out that all sharp games are improving, which simplifies our methods compared to the conference version and thus improves the results.

\bibliographystyle{alphaurl}
\bibliography{bib}

\end{document}